\documentclass[aps,amsmath,twocolumn,amssymb]{revtex4}
\usepackage{graphicx}
\usepackage{dcolumn}
\usepackage{bm}
\usepackage{amsmath,amsthm}
\usepackage{hyperref}
\usepackage{color}
\usepackage{enumerate}

\newtheorem{theorem}{Theorem}
\newtheorem{lemma}{Lemma}
\newtheorem{definition}{Definition}

\begin{document}

\def\ket#1{|#1\rangle}
\def\bra#1{\langle#1|}
\newcommand{\ketbra}[2]{|#1\rangle\!\langle#2|}
\newcommand{\braket}[2]{\langle#1|#2\rangle}
\newcommand{\note}[1]{({\bf note: #1})}
\newcommand{\prob}[1]{{\rm Pr}\left(#1 \right)}
\newcommand{\Tr}[1]{{\rm Tr}\!\left[#1 \right]}
\newcommand{\expect}[2]{{\mathbb{E}_{#2}}\!\left\{#1 \right\}}
\newcommand{\var}[2]{{\mathbb{V}_{#2}}\!\left\{#1 \right\}}
\newcommand{\change}{}

\newenvironment{proofof}[1]{\begin{trivlist}\item[]{\flushleft\it
Proof of~#1.}}
{\qed\end{trivlist}}

\newcommand{\cu}[1]{{\textcolor{red}{#1}}}
\newcommand{\tout}[1]{{}}
\newcommand{\beq}{\begin{equation}}
\newcommand{\eeq}{\end{equation}}
\newcommand{\beqa}{\begin{eqnarray}}
\newcommand{\eeqa}{\end{eqnarray}}

\newcommand{\id}{\openone}

\newcommand{\observ}{A_i}
\bibliographystyle{apsrev}
\newcommand{\ensemble}{E_H}


\title{Information-theoretic equilibration:  \\
 the appearance of irreversibility under complex quantum dynamics}
\author{Cozmin Ududec$^{1,2}$,  Nathan Wiebe$^{2,3}$, and Joseph Emerson$^{2,4}$
}
\affiliation{$^1$ Department of Physics and Astronomy, University of Waterloo, Waterloo, ON, Canada}
\affiliation{$^2$ Institute for Quantum Computing, 200 University Ave West, Waterloo, ON, Canada}
\affiliation{$^3$ Department of Combinatics \& Opt., University of Waterloo, Waterloo, ON, Canada}
\affiliation{$^4$ Department of Applied Math, University of Waterloo, Waterloo, ON, Canada}

\begin{abstract}
The question of how irreversibility can emerge as a generic phenomena when the underlying mechanical theory is reversible has been a long-standing fundamental problem for both classical and quantum mechanics. 
We describe a mechanism for the appearance of irreversibility that applies to coherent, isolated systems in a pure quantum state. This equilibration mechanism requires only an assumption of sufficiently complex internal dynamics and natural information-theoretic constraints arising from the infeasibility of collecting an astronomical amount of measurement data. Remarkably, we are able to prove that irreversibility can be understood \emph{as typical} without assuming decoherence or restricting to coarse-grained observables, and hence occurs under distinct conditions and time-scales than those implied by the usual decoherence point of view. We illustrate the effect numerically in several model systems and prove that the effect is typical under the standard random-matrix conjecture for complex quantum systems. 
\end{abstract}

\maketitle




There has been considerable recent interest in the sufficient conditions for equilibration~\cite{PSW05,GoldsteinOnvN,Rigol,Zanardi,LPS+09,Brandaoetal,MRA11,SethThesis,Goldstein06,GemmerBook, Znidaric11,Rei10,Sho11,ShoFarr12,ReKa12}. 
   These approaches normally assume a decoherence mechanism resulting from the entanglement between the system of interest and a larger environment, or else assume highly coarse-grained observables.   In this work we describe a mechanism for equilibration 
that applies to isolated quantum systems in pure states, without assuming decoherence, restricting to subsystems, time-averaging or coarse-graining the observables. 
The mechanism for equilibration that we describe is an \emph{information-theoretic} one  that requires an assumption of complex internal dynamics  coupled  with realistic limitations to predicting the detailed evolution of the system and the experimental infeasibility of collecting an astronomically large amount of measurement data. This approach builds on earlier arguments by Peres \cite{Peres84} and Srednicki \cite{Srednicki1,Srednicki2} who proposed that the statistical complexity of the system's eigenvectors could be responsible for equilibration in isolated quantum systems.  We show that these conditions are sufficient to account for the effective (microcanonical) equilibration of the measurement statistics for natural choices of (non-degenerate) observable, meaning that, after a finite equilibration time, the dynamical state becomes effectively indistinguishable from the microcanonical state.
Hence information-theoretic equilbration (ITE) accounts for microcanonical equilibration in a way that is directly analogous to how classical chaos (mixing) accounts for the microcanonical equilibration of classically chaotic systems \cite{Gaspard,Dorfman}.
Remarkably, we are able to prove that ITE is universal for complex systems under the standard random-matrix conjecture~\cite{Meh04,Haa06}. Specifically, we prove that information-theoretic equilibration occurs with high probability for individual Hamiltonians drawn from two physically relevant ensembles:  the Gaussian Unitary Ensemble (GUE), which has a succesful history predicting unversal features of complex quantum systems~\cite{Meh04}, and a random local Hamiltonian (RLH) ensemble consisting of many-body systems restricted to two-body interactions. We then illustrate ITE numerically  in some surprisingly simple examples of Hamiltonian models under natural choices of (maximally fine-grained) observable: a two-field variant of the many-body Heisenberg Hamiltonian as well as the quantum kicked top~\cite{Haa06}, which is a single-body, classically chaotic system.  

Consider a pure state evolving under a Hamiltonian $H$, $\rho(t) = \exp(-iHt/\hbar) | \psi(0) \rangle \langle \psi(0) |  \exp(iHt/\hbar)$.  The dynamical state $\rho(t)$ can not reach the true equilibrium state $\sigma_\infty := \lim_{\tau \rightarrow \infty } \frac{1}{\tau} \int_0^\tau \rho(t) dt$~\cite{PSW05,Sho11} because the state remains pure. In particular,  the trace distance $\|\rho(t)-\sigma_{\infty}\|$, which characterizes the distinguishabiity under an \emph{optimal} choice of measurement operator, can be large throughout the evolution.  However, for a given complex system $H$ in a large Hilbert space, even a suboptimal measurement that enables distinguishability of these two states at any time $t$ may neither be known theoretically nor easily engineered experimentally.  For example, for a cubic lattice of dipolarly coupled spins, which is an analytically intractable system that has been probed experimentally for decades, only recently was a measurement procedure devised that revealed long-lived (multiple-quantum) coherence after equilibration of the free-induction decay~\cite{CLB05}.
Conceptually then we see that the appearance of equilibration can and does result from insufficient knowledge of, or control over, choice of observable. Our contribution is to characterize and illustrate conditions under which the signatures of purity and coherence are provably ``lost in Hilbert space'', and hence unobservable due to realistic limitations on both theoretical and experimental abilities. 

 We remark that our assumptions are conceptually similar and yet distinct from those of the usual decoherence argument, in which a system coupled to a reservoir appears to reach equilibrium (due to  entanglement between the system and reservoir) although the joint state of system plus reservoir remains pure.  That conclusion holds only if one assumes that one can not predict or perform the kind of (entangling) measurement across the combined system plus reservoir that would readily distinguish the actual state from the equilibrium one; that is, the argument goes through by restricting the set of observables to local ones. In contrast, our observation is that information-theoretic limitations alone are sufficient to account for the appearance of equilibration for accessible obervables on complex systems and so, contrary to the usual assumption (see \cite{HSZ98,Zurek98,PSW05}),  decoherence from a reservoir is not necessary from an explanatory point of view.   More practically, whereas the time-scale for equilibration under decoherence depends on the strength of the coupling to the reservoir, our mechanism does not and predicts equilibration on a distinct, and potentially shorter, time-scale. 
Furthermore, our approach is a natural quantum analog of classical \emph{microcanonical} equilibration \cite{Gaspard,Dorfman}.


We consider a quantum system with some kinematically accesible Hilbert space that is finite-dimensional $\mathcal{H} = \mathcal{C}^D$.  In order to show that we do not require coarse-graining, we consider a  maximally fine-grained (ie, non-degenerate) observable $A$ acting on $\mathcal{H}$,  where $A = \sum_{k=1}^D a_k \hat{P}_k$  with rank-one orthogonal projectors $\hat{P}_k$. Our argument applies also to local or other coarse-grained observables (which can be represented by degeneracies).  For simplicity of analysis we consider the (most adverserial) setting where the system starts in a pure state that is maximally localized with respect to $A$, ie, $\rho_0 = | a_i \rangle\! \langle a_i| $, and then examine how the pure states spreads out over the eigenbasis of $A$ under time-evolution given a Hamiltonian $H$.
The empirical question of whether the system appears to approach (microcanonical) equilibrium given some observable $A$ corresponds to asking whether the experimental measurement statistics for the evolved pure state can be distinguished from those of the equilibrium state. Hence the relevant quantities for this task are the probabilities over distinct outcomes $k$, 
\beq
\prob{k | \rho(t)} =
\Tr{\hat{P}_k U(t) \rho_0  U^\dag(t) },\label{eq:pr}
\eeq
 and the goal is to distinguish $\rho(t)$ from  $\sigma_\infty$ by sampling the distribution in~\eqref{eq:pr}.  For simplicity we focus on cases where $\sigma_\infty = \id/D  $, but  $\sigma_\infty$ may differ from the micro-canonical state $\rho_{\rm mc} := \id/D$ or any thermal state \cite{Rigol}.

\begin{definition}\label{info theor equil}
\change A Hamiltonian $H$ acting on $\mathcal{H}=\mathbb{C}^D$ exhibits \emph{information theoretic equilibration} (ITE)  with respect to  an observable $A$ at a time $t$,  if the outcome distribution $\prob{k | \rho(t)}$ can only be distinguished from the micro-canonical distribution ${\rm{Pr}_{mc}}(k)$ with probability at least $1-O(1/\rm{poly}(D))$ by (a) taking a number of samples from $\prob{k|\rho(t)}$ that scales at least as $O({\rm poly}(D))$ or (b) performing any information processing that requires at least $O({\rm poly}(D))$ arithmetic or logical operations.
\end{definition}

This definition emphasizes that although the exact quantum distribution for the system may be \emph{in principle} distinguishable from the microcanonical distribution, the two are effectively indistinguishable if the resources needed to distinguish them exceed those practically available. We delineate the practical from the impractical  by disallowing resources (the number of measurements taken and computational time used in their analysis) that grow polynomially with the Hilbert space dimension (and hence exponentially with the number of subsystems). Of course, for a different physical scenario, a different cut-off may be appropriate.  Our condition (b) includes a restriction on computational resources because the two distributions could be distinguished using fewer samples if the $\prob{k|\rho(t)}$ can be pre-computed.  In other words,  information-theoretic equilibration is relevant precisely when the system is in a sufficiently large Hilbert space that such a pre--computation is infeasible.  We represent our ignorance of $\prob{k|\rho(t)}$ by assuming that it is drawn from a distribution that is invariant under permutations of outcome labels.
\change We now show in the following theorem that, without the ability to efficiently predict $\prob{k|\rho(t)}$, ITE with respect to a particular measurement occurs when the \emph{outcome variance},
\begin{equation}
\var{\prob{k}}{k} :=  D^{-1} \sum_{k=0}^{D-1}  \left[ \prob{k} - D^{-1}  \right]^2,\label{eq:varprob}
\end{equation}
 is sufficiently small, which is typical of cases where the underlying dynamics has no constants of motion.  Proof is provided in the supplemental material.

\begin{figure}[t!]
\includegraphics[width=\linewidth]{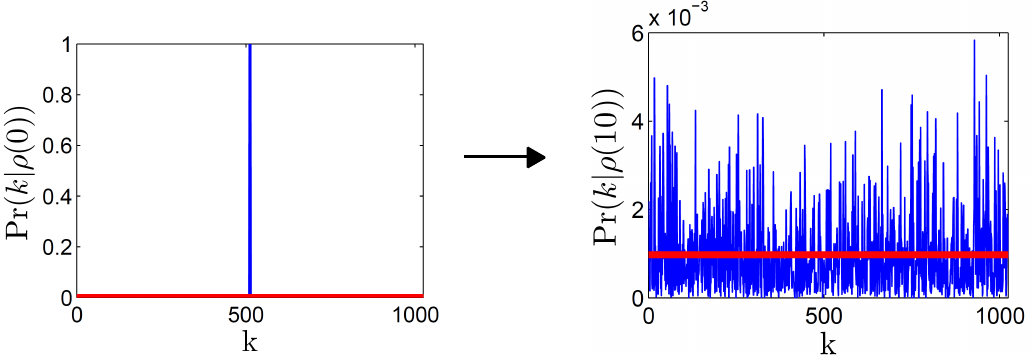}
\caption{\change We plot $\prob{k|\rho(t)}$ for Pauli--Z measurements given by quantum theory (blue) and the uniform distribution (red) for a random local Hamiltonian acting on $10$ qubits at $t=0$ and for $t>t_{\rm eq}$.  ITE arises  from the difficulty in distinguishing  quantum fluctuations from sampling errors for $t>t_{\rm eq}$.\label{fig:equib}}
\end{figure}

\begin{theorem}\label{thm:distinguish} \change
Consider an unknown distribution that is promised to be with equal probability either  (a)  the uniform distribution on the set $\mathcal{S}=\{1,\ldots,M\}$ or (b) an unknown distribution $P(k)$ that is drawn from a distribution over probability distributions on $\mathcal{S}$ with outcome variances that scale as $O(M^{-2})$ such that $\prob{P({k})}$ is invariant with respect to permutations of $\mathcal{S}$.  With high--probability, the probability of correctly distinguishing between (a) and (b) after obtaining $N$ samples is at most $1/2 + O(N/M^{1/4})$.
\end{theorem}

\change Theorem~\ref{thm:distinguish} shows that  $N\ge O(M^{1/4})$ samples are needed to distinguish the distributions with probability substantially greater than $1/2$, which is prohibitively expensive in the case of a non--degenerate projective measurement because $M=D$.  Similarly, if we consider a generalized measurement with $M> D$ (as is relevant in the case of SIC POVMs), Theorem~\ref{thm:distinguish} similarly shows that distinguishing the distributions is hard.  Finally,  it is straightforward to show that coarse grained measurements with $M<D$  do not provide an advantage  under the assumptions of Theorem~\ref{thm:distinguish} because the permutation invariance of the prior distribution over $\prob{k}$ prevents such strategies from succeeding with high probability.   Another consideration is that $\mathbb{V}_k\{\prob{k}\}= O(M^{-2})$ does not imply that the fluctuations  are negligible \emph{in principle}; in fact, it is consistent with $\|\prob{k|\rho(t)}-1/M\|_1$ being constant, which implies that an optimal measurement exists that can distinguish the two distributions efficiently~\cite{nielsen}.  Hence Theorem~\ref{thm:distinguish}  is only meant to give a hardness result for distinguishing two states  given the induced distributions with respect to a fixed measurement, and does not apply to cases where the optimal measurement is both known a priori and experimentally accessible. Indeed the exceptions to our assumptions are relevant, e.g., when the system admits constants of the motion that are simple relative to the selected observable.

Which Hamiltonian systems satisfy the assumptions of Theorem \ref{thm:distinguish}, for natural choices of $A$, and hence exhibit information theoretic equilibration? Pure-state fluctuations satisfying the scaling of Theorem \ref{thm:distinguish} were observed already in the two-body, classically chaotic quantum system studied in Refs.~\cite{EB01,EB01b}, which motivated the question: was the behaviour of that complex system exceptional, or was it evidence of a universal equilibration behaviour for closed chaotic systems? If the latter, does this effect carry over from chaotic quantum systems to sufficiently complex many-body quantum systems? 

To answer these questions, we take the enormously succesful approach of Wigner and Dyson and the army of theoretical physicists following them who have demonstrated that certain features of appropriate random matrix ensembles (RME) can predict typical properties of complex quantum systems.  This is known as the random-matrix conjecture, and it has provided accurate predictions of the spectral properties of heavy nuclei \cite{Meh04}, spectral and eigenvectors statistics of quantum chaos models \cite{KMH88,KZ91}, and quantum transport in mesoscopic structures \cite{Beenakker}. 
Consider any ensemble that has a mean that equilibrates information theoretically with respect to $A$ and is sufficiently sharply peaked about that mean, then individual systems from the ensemble will satisfy Theorem~\ref{thm:distinguish} with (very) high probability.   This phenomenon, known as \emph{concentration of measure}, is central to the random-matrix conjecture, and it is important to note that our averages over the ensemble are not an implicit appeal to decoherence or mixing, but a  method for estimating the \emph{typical properties} of \emph{individual systems} within the ensemble. 

\change The system must be allowed to evolve for a sufficient amount of time for the state to spread out
from a distribution with support on initial eigenstate of $A$ to one that obeys $\prob{k|\rho(t)}\approx 1/D$ for our result to hold (see~Fig.~\ref{fig:equib}).
We refer to the earliest such time as the equilibration time, which we denote $t_{\rm eq}$.  For an individual system, we also require that for most $t\ge t_{\rm eq}$ that $\prob{k|\rho(t)}$ is nearly maximally spread out.  If the Hamiltonian is drawn from an ensemble, it is then possible to define an equilibration time such that almost all Hamiltonians drawn from the ensemble achieve ITE with respect to $A$ and $t\ge t_{\rm eq}$:


\begin{figure}[t!]
\includegraphics[width=\linewidth]{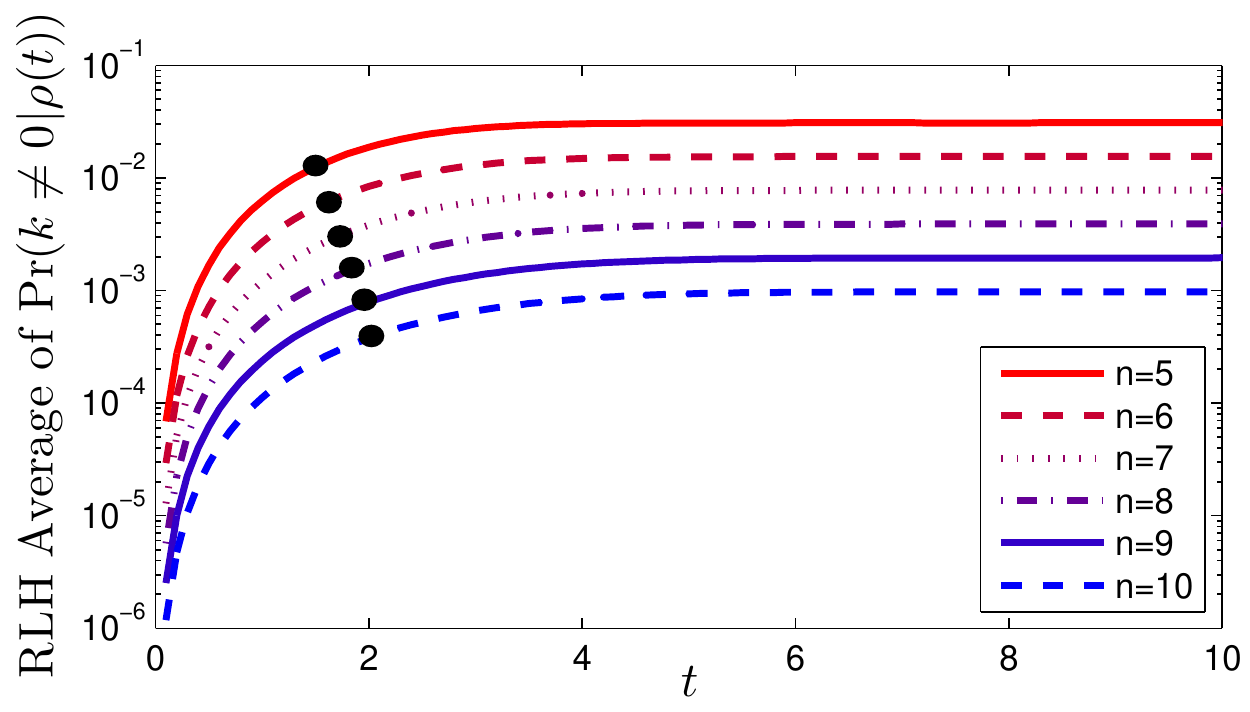}
\caption{The RLH ensemble average of the probabilities $\prob{k\ne 0|e^{-iHt}\ketbra{0}{0}e^{iHt}}$ of the evolved state for $250$ random Hamiltonians  plotted as a function of time for $5$--$,\ldots,10$--qubit systems.  The circles show $t_{\rm eq}$ for each $n$, which scale roughly as $O({\log(D)})$.  
} \label{fig:RLHeq}
\end{figure}

\begin{lemma}\label{lem:concentrate}
\change
 Almost all Hamiltonians sampled from an ensemble of Hamiltonians equilibrate information theoretically with respect to a fixed observable $A$  and time $t>t_{\rm eq}$, in the limit as $D\rightarrow \infty$, if the ensemble average and variance (denoted $\mathbb{E}_{\ensemble}$ and $\mathbb{V}_{\ensemble}$ respectively) of the outcome variance obey for all $t\ge t_{\rm eq}$
\begin{align}
\mathbb{E}_{\ensemble}\!\left\{\mathbb{V}_k\{\prob{k|\rho(t)}\} \right\}&\le O(D^{-2})\label{eq:conc1}\\
\mathbb{V}_{\ensemble}\!\left\{\mathbb{V}_k\{\prob{k|\rho(t)}\}\right\}&\le O(D^{-4})\label{eq:conc2}.
\end{align}
\end{lemma}

Proof is given in the supplemental material. 

We now give our first evidence for universality by proving that ITE is typical for the important Gaussian Unitary Ensemble (GUE),  which defines an invariant measure on the set of Hamiltonians. The GUE  is the appropriate model a highly successful model for many properties of complex physical systems with no hidden symmetries \cite{Meh04}. 

\begin{theorem}\label{thm:gue}
Take a non-degenerate observable $A$ acting on $\mathcal{H}=\mathbb{C}^D$, and an initial pure state $\rho_0 = \ket{x}\!\bra{x}$ which is an eigenstate of $A$. Almost all Hamiltonians drawn from GUE then equilibrate information theoretically with respect to $A$ and $t\ge t_{\rm eq}$ in the limit as $D\rightarrow \infty$ for $t_{\rm eq}(D)=O(D^{-1/6})$.
\end{theorem}

The proof is in the supplemental material. This theorem tells us the remarkable result that, as $D$ increases,  the overwhelming majority of Hamiltonians will cause an initially pure, localized state to spread out over the non-degenerate eigenbasis of  $A$ in a sufficiently uniform manner, to become practically indistinguishable from the microcanonical state for any $t\ge t_{\rm eq}$.
Thanks to decades of numerical studies of GUE as a model of complex many-body systems \cite{Beenakker} and few-body quantum chaos systems \cite{Haa06}, it is known that GUE is a good predictor of short-range spectral fluctuations  \cite{BGS84},  and low-order moments of eigenvector components \cite{KMH88,KZ91}, but not a good predictor of long-range spectral fluctuations \cite{Haa06}. Our proof of the smallness of the fluctuations using GUE (for $t> t_{\rm eq}$) depends only on low-order moments of the eigenvector components, i.e., unitary $t$-design condition with $t$=8 \cite{Dankertetal}  (see supplementary material for details).  Hence we expect this  aspect of the GUE model to be reflected in physically relevant  Hamiltonian systems. However, we do not expect the GUE prediction for the equilibration time-scale to be physically relevant (clearly the value of $t_{\rm eq}$ for GUE is unrealistically short) because it depends on long-range spectral fluctuations.  We now confirm both of these expectations  for two RMEs consisting of many-body spins with two-body interactions, and conclude by demonstrating ITE with respect to \change tensor product measurements on a physically relevant time-scale in some example model systems.

\begin{figure}[t!]
\includegraphics[width=\linewidth]{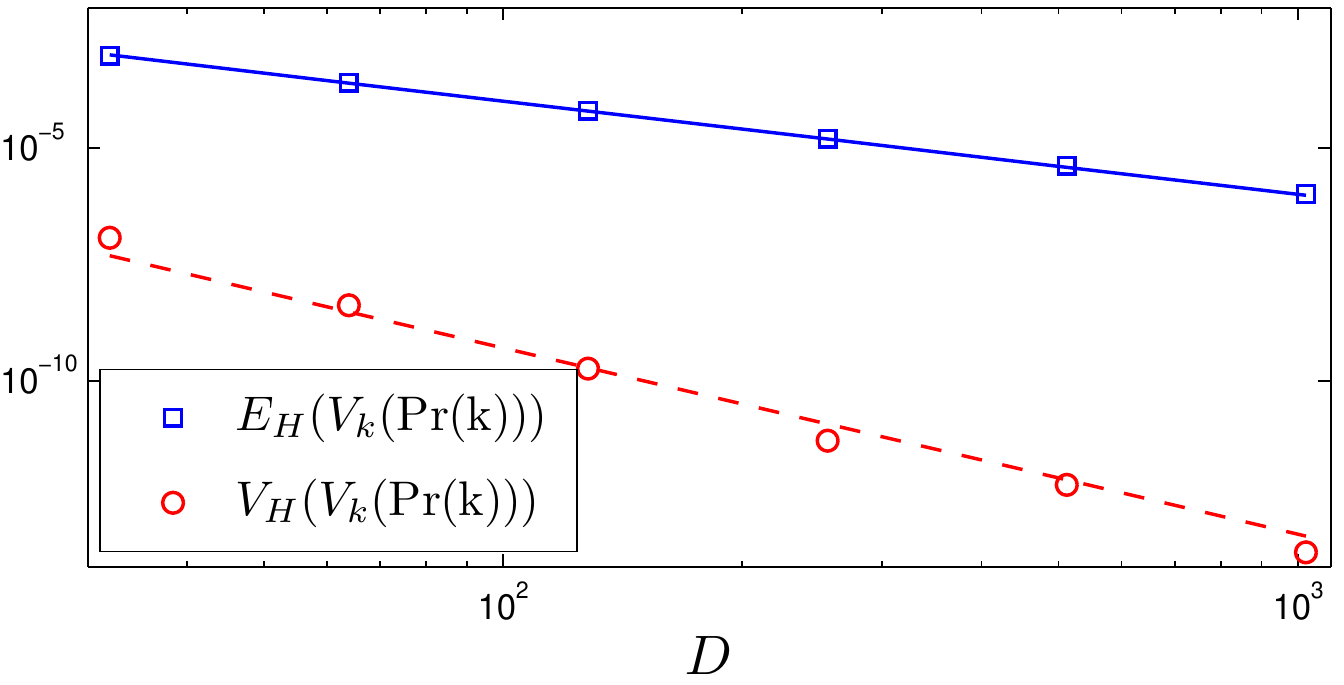}
\caption{Numerically computed expectation values and variances over the RLH ensemble of the outcome variance computed at $t=10$ where $t_{\rm eq}\lessapprox 2$ for $\rho(0)=\ketbra{0}{0}$.  The data was obtained for $250$ randomly chosen Hamiltonians, and shows that $\mathbb{V}_{E_H}(\mathbb{V}_k(\prob{k|\rho(t)})))\approx 0.05 D^{-4}$ and $\mathbb{E}_{E_H}(\mathbb{V}_k(\prob{k|\rho(t)}))\approx D^{-2}$. \label{fig:RLH}}
\end{figure}

We construct an ensemble of random local Hamiltonians (RLH) on $n$ spins, consisting of $2$--body interactions  between $2$--level quantum systems, as follows:
\begin{equation}
H=\|H\|^{-1}\left(\sum_{i=1}^n\sum_{p} a_{i,p}\sigma_p^{(i)}+ \sum_{i<j}\sum_{p,p'} b_{i,j,p,p'}\sigma_p^{(i)}\sigma_{p'}^{(j)}\right), \nonumber
\end{equation}
where $p,p' \in \{X,Y,Z\}$, and each $a_{i,p}$ and $b_{i,j,p,p'}$ is a Gaussian random variable with mean $0$ and variance $1$.  
We consider the observable $A= \sum_{j=0}^{D-1} a_j \ket{j}\!\bra{j}$, corresponding to a non-degenerate projective measurement in the eigenbasis of $\sigma_z^{\otimes n}$.  
RLH is clearly invariant under permutation of qubit labels and local rotations of each qubit, and therefore our results also apply to any $A'$ that differs from $A$ by local rotations.
Fig.~\ref{fig:RLHeq} shows that pure states evolving under individual elements of RLH approach equilibrium as $D$ increases.
We estimate the equilibration time using the location of the inflection points of the curves in Fig.~\ref{fig:RLHeq}, and find it scales as $O({\log(D)})$, which is characteristic of quantum chaotic systems~\cite{EB01,EB01b,Haa06}.
Fig.~\ref{fig:RLH}  shows that the outcome variance for a typical Hamiltonian chosen uniformly from the RLH ensemble satisfies the requirementes of Lemma~\ref{lem:concentrate}, which implies that almost all RLH Hamiltonians will equilibrate information theoretically with respect to any non-degenerate measurement in the class $A'$ as $D\rightarrow \infty$ for any $t\ge t_{\rm eq}$. 
We further strengthen the physical relevance of this result by showing that ITE still holds for \change $t\ge t_{\rm eq}$  when the 2-local Hamiltonians are constrained to have nearest-neighbor interactions in one-- and two--dimensions (see supplemental material). 

\begin{figure}[t!]
\includegraphics[width=0.95\linewidth]{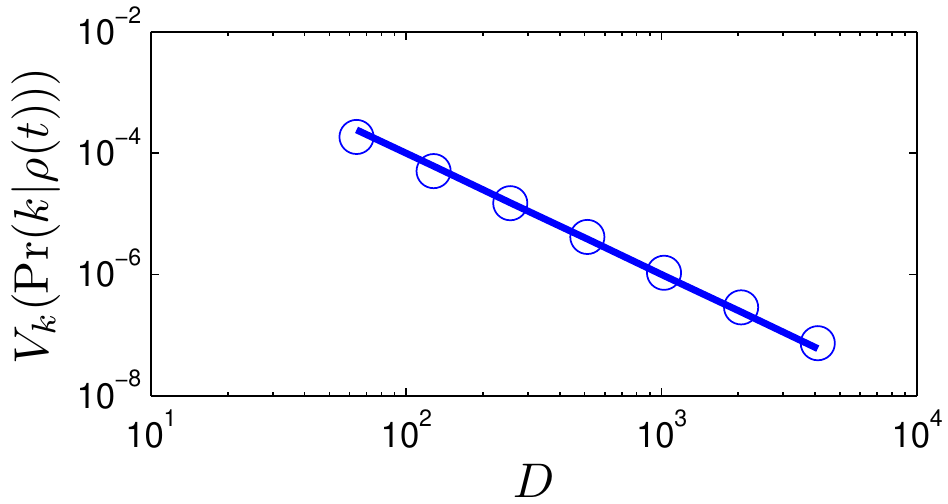}
\caption{Evidence of ITE for \change $t\ge t_{\rm eq}$ in an extremely simple  many-body system with nearest-neighbor interactions for increasing number of spins $n$ ($D=2^n$). The measurement consists of readout of each spin along the z-axis. The plot shows $\mathbb{V}_k\{\prob{k|\rho(t)}\}\approx 1.6D^{-2}$ for the Hamiltonian \eqref{eq:genheis} with $t=20$ where $t_{\rm eq}\simeq15$  (see supplemental material).}
\label{fig:vscaleH}
\end{figure}

We now give two simple examples of individual model systems that exhibit information theoretic equilibration: a many-body system that is a two-field variant of the Heisenberg Hamiltonian  and a one-body chaotic model, the quantum kicked top.  
The two-field variant of the Heisenberg mode consists of $n$-spins arranged in a line with periodic boundary conditions:
\begin{equation}\label{eq:genheis}
H=\frac{1}{\|H\|}\! \! \left( \sum_{i\le n/2} \sigma_z^{(i)}\!+\!\! \! \sum_{i> n/2} \! \! \! \sigma_x^{(i)}\!+\! \sum_i \vec{\sigma}^{(i)}\cdot \vec{\sigma}^{(i+1)}\right) \! \! .
\end{equation}
We choose this Hamiltonian because it is highly structured local Hamiltonian that is not typical of RLH and yet it is unstructured enough to be non-integrable so there are no constants of motion that prevent equilibration on the full Hilbert space (otherwise ITE would be limited to the invariant subspaces fixed by the constants of motion). 
Figure~\ref{fig:vscaleH} shows that the outcome variance of the probability distribution indeed scales as $O(D^{-2})$ with respect to $A=\sum_{j=0}^{D-1} a_j \ketbra{j}{j}$, corresponding to read-out of all spins in the computational basis.  Hence Theorem~\ref{thm:distinguish} implies that information theoretic equilibration occurs for this simple many-body Hamiltonian with respect to a natural observable.  This is evidence that our equilibration mechanism is not just a mathematical feature of random Hamiltonian ensembles but occurs also in a simple, physically accessible many-body model.   
We also demonstrate ITE for the quantum kicked top in a regime of global chaos \change with respect to non-degenerate measurements in the $J_z$  basis (see supplemental material) and physically accessible times. 


\emph{Conclusion.} We have demonstrated a novel mechanism for equilibration that holds very broadly for the probability distributions of even maximally fine-grained measurements on pure quantum states of closed Hamiltonian systems.
Remarkably, this information theoretic equilibration is observed to hold without requiring any form of decoherence or restricting to local or otherwise coarse-grained measurements. This is because, in the typical case of a complex system, the dynamical pure-state quantum fluctuations, though finite, do \emph{not} lead to a \emph{breakdown of correspondence} with the equilibrium state (contrary to a common implicit assumption, see Refs.~\cite{HSZ98,Zurek98,PSW05}) because they become unobservably small under purely statistical considerations  (in the limit of large $D$) after the equilibration time-scale.
Our key insight is that although dynamical pure states of complex systems exhibit coherent fluctuations away from true micro-canonical equilibrium,  their detection in practice  requires extraordinary experimental resources, such as collecting  $O(D^{1/4})$ measurement outcomes from repetitions of the experiment, or pre--computation of the location of the dynamical state in a $D$-dimensional Hilbert space, or performing joint (entangling) measurements on identical copies of the system. In the absence of such resources, by Theorem~\ref{thm:distinguish} we see that after some finite time, the empirical probability distributions for dynamical pure states of complex quantum systems cannot be distinguished from the micro--canonical equilibrium state.

\appendix
\setcounter{lemma}{1}

\section{Proof of Theorem 1}\label{AppII}
\begin{proofof}{Theorem~1}
The first step of the proof is to demonstrate that if w1e take $N\ll \sqrt{M}$ samples from the uniform distribution then the probability of obtaining $N$ distinct
outcomes is nearly $1$.  Since there are $M!/(M-N)!$ ways $N$ unique items can be selected from a set of $M$ and since there are $M^N$ possible selections of items, we have that the probability of seeing no coincidences if the actual distribution is the uniform distribution is
\begin{equation}
\frac{M!}{(M-N)!M^N}= 1-\frac{N(N-1)}{2M}+O(1/M^2).
\end{equation}
This implies that coincidental outcomes are unlikely unless $N\ge O(\sqrt{M})$.

Now let us assume that the true probability distribution is not uniform, but rather a distribution with outcome variance $\var{\prob{k}}{k}=O(M^{-2})$.  We compare this to a distribution with
$\var{{\rm{Pr}}(k)}{k}=O(M^{-2})$ that has the highest probability of coincidence for some outputs.
Using the the definition of variance, $\var{\prob{k}}{k}=M^{-1}\sum_{k}(\prob{k}-M^{-1})^2$, expanding the sum and using $\sum_k \prob{k}=1$, we find
\begin{equation}
\var{\prob{k}}{k}=O(M^{-2})\ \Rightarrow {\prob{k}}\le O(M^{-1/2}), \ \forall~k.
\end{equation}

The probability that no coincidental measurements are observed after $N$ measurements is therefore at least
\begin{equation}
\prod_{j=1}^{N-1}\left(1-O\left(\frac{j}{M^{1/2}}\right)\right)=1-O\left(\frac{N^2}{\sqrt{M}}\right).
\end{equation}

  Let the event $\mathcal{M}$ denote the observation that $N$ unique measurement outcomes are observed, $\bar{\mathcal{M}}$ be the event where at least one measurement outcome is repeated,  $m_1$ be the model that prescribes a uniform probability distribution to the outcomes, $m_2$ be the model that the probability distribution has outcome variance $O(M^{-2})$, and $\mathcal{D}$ be the sequence of samples yielded by the device.  

\change Since the underlying distribution $P(k)$ is drawn from a distribution over distributions on $\mathcal{S}=\{1,\ldots,M\}$ that is invariant under permutations of $\mathcal{S}$, $\prob{\mathcal{D}|m_2}=\prob{\mathcal{P}(\mathcal{D})|m_2}$, where $\mathcal{P}$ is a permutation of $\mathcal{S}$.  We therefore see that all sequences of measurement outcomes that are equivalent up to permutations of labels provide equivalent evidence for model $m_2$.  Since $m_1$ is the uniform distribution, $\prob{\mathcal{D}|m_1}=\prob{\mathcal{P}(\mathcal{D})|m_1}$ for any permutation $\mathcal{P}$. It is then clear that the labels of
 the outcomes observed cannot be used to distinguish between model $m_1$ and $m_2$.  We therefore can, without loss of generality, choose the label of the outcomes such that the first outcome observed is outcome $1$, the second unique outcome observed is $2$ an so forth.  From this perspective, it is clear that the differences between both models only become apparent in the distribution of coincidental outcomes.  Our proof then follows from Bayes' theorem and by showing that the probability of a coincidental outcome is small unless $N\ge O(M^{1/4})$.  

Note that the preceding argument effectively prevents coarse graining from allowing us to distinguish the two distributions with
high probability using a small number of measurements because the probability of correctly guessing a coarse graining that assigns high probability to particular coarse--grained outcomes is $O(1/{\rm poly}(M))$ (given the assumption of permutation invariance).

\change There are two possible scenarios: either $\mathcal{D}$ does not contain any repeated sample labels or $\mathcal{D}$ contains at least one repeated sample label.  We will first assume that the entries of $\mathcal{D}$ are unique, which we denote event $\mathcal{M}$.   Bayes' Theorem then implies
\begin{equation}
\prob{m_1|\mathcal{M}}=\frac{\prob{\mathcal{M}|m_1}\prob{m_1}}{\prob{\mathcal{M}|m_1}\prob{m_1}+\prob{\mathcal{M}|m_2}\prob{m_2}}.\label{eq:bayes}
\end{equation}

From our previous discussion, we see that
\begin{align}
|1-\prob{\mathcal{M}|m_1}|&=O\left(\frac{N^2}{M}\right)\nonumber\\
|1-\prob{\mathcal{M}|m_2}|&\le O\left(\frac{N^2}{\sqrt{M}}\right).
\end{align}
These results, and the fact that $\prob{m_1}+\prob{m_2}=1$ give us 
\begin{equation}
\prob{m_1|\mathcal{M}}\le\frac{\left(1-O\left(\frac{N^2}{M}\right)\right)\prob{m_1}}{1-O\left(\frac{N^2}{\sqrt{M}}\right)}.\label{eq:a5}
\end{equation}
We apply Taylor's Theorem to the denominator of~\eqref{eq:a5} and find that 
\begin{equation}
|\prob{m_1|\mathcal{M}}-\prob{m_1}|\le O\left(\frac{N^2}{\sqrt{M}}\right),\label{eq:condmbound}
\end{equation}
and similarly
\begin{equation}
|\prob{m_2|\mathcal{M}}-\prob{m_2}|\le O\left(\frac{N^2}{\sqrt{M}}\right).
\end{equation}
This shows us that the support provided by event $\mathcal{M}$ for either hypothesis  is small unless $N\ge O(M^{1/4})$.

Our next step is to formally show that a typical data set $\mathcal{D}$ will be, with high probability, uninformative unless $N\ge O(M^{1/4})$.  This follows from a concentration of measure argument over $\mathcal{D}$ for the posterior probability distribution.
The average over $\mathcal{D}$ of the posterior probability is
\begin{equation}
\mathbb{E}_{\mathcal{D}}(\prob{m_1|\mathcal{D}})\!=\!\prob{m_1|\mathcal{M}}\!\prob{\mathcal{M}}\!+\!\prob{m_1|\bar{\mathcal{M}}}\!\prob{\bar{\mathcal{M}}}\!.\label{eq:Edbound}
\end{equation}
\change Using $\prob{m_1|\bar{\mathcal{M}}}\le 1$, $|1-\prob{\mathcal{M}}|\le O(N^2/\sqrt{M})$ and $\prob{\bar{\mathcal{M}}}\le O(N^2/\sqrt{M})$ we 
find from~\eqref{eq:condmbound} that	~\eqref{eq:Edbound} implies
\begin{equation}
|\mathbb{E}_{\mathcal{D}}(\prob{m_1|\mathcal{D}})-\prob{m_1}|\le O(N^2/\sqrt{M}).\label{eq:expectD}
\end{equation}
A similar calculation gives the variance as
\begin{align}
\mathbb{V}_{\mathcal{D}}(\prob{m_1|\mathcal{D}})&=\mathbb{E}_{\mathcal{D}}(\prob{m_1|\mathcal{D}}^2)\!-\!(\mathbb{E}_{\mathcal{D}}(\prob{m_1|\mathcal{D}}))^2\nonumber\\
&\le O(N^2/\sqrt{M}).
\end{align}
Therefore, Chebyshev's inequality implies that the probability of a given data set $\mathcal{D}$ deviating substantially from
the expectation value is
\begin{equation}
\prob{|\prob{m_1|\mathcal{D}}-\mathbb{E}_{\mathcal{D}}(\prob{m_1|\mathcal{D}})|\ge \epsilon}\le \frac{N^2}{\epsilon^2\sqrt{M}}.
\end{equation}
Using this result in concert with~\eqref{eq:expectD} implies that, with high probability, the posterior probability distribution after taking $N$ samples will obey
\begin{equation}
|\prob{m_1|\mathcal{D}}-\prob{m_1}|\le O\left(\frac{N}{M^{1/4}} \right).
\end{equation}
The result of Theorem~1 then follows from choosing the prior $\prob{m_1}=1/2$.

\end{proofof}

\section{Proof of Lemma~1}
\begin{proofof}{Lemma~1}
Eqns.~(3),~(4) and Chebyshev's inequality imply that in the limit of large $D$, the outcome variance of 
$\prob{k | \rho(t)}$ for $t\ge t_{\rm eq}$ is concentrated around $O(D^{-2})$.  In particular, for any $\epsilon>0$ we have from Chebyshev's inequality and~(4) that
\begin{align}
&{\underset{H \sim \mu_{D}}{\rm Pr}} \left( \left| \mathbb{V}_k\{\prob{k | \rho(t)}\} - \mathbb{E}_{\ensemble}\{\mathbb{V}_k\{\prob{k | \rho(t)}\} \} \right|  \geq \epsilon \right)\nonumber\\
 &\qquad\le\frac{ O(D^{-4})}{\epsilon^2},\label{eq:chebyprob}
\end{align}
where $\mu_D$ is the appropriate measure for the ensemble $\ensemble$.  In other words, the outcome variance for an individual system are very close to the ensemble average.

We can see that almost all Hamiltonians will have outcome variance $O(D^{-2})$ in the limit of large $D$ via the following argument.  To compute the probability that the outcome variance of a particular Hamiltonian scaling as $O(D^{-2+\gamma})$ for some $\gamma>0$, we set $\epsilon= O(D^{-2+\gamma})$.  Equation~\eqref{eq:chebyprob} then implies that the probability of such an event scales at most as $O(D^{-2\gamma})$, which vanishes in the limit of large $D$ unless $\gamma=0$.  Almost all Hamiltonians chosen from the ensemble therefore have outcome variance $O(D^{-2})$ in the limit of large $D$ if $t\ge t_{\rm eq}$.

\change Next we will show that this implies information theoretic equilibration with respect to the observable $A$ and time $t\ge t_{\rm eq}$.  We know that almost all Hamiltonians drawn from the ensemble $\ensemble$ will satisfy the requirements of Theorem 1.  Let us then consider a decision problem where we are maximally ignorant whether the state is a distribution that is uniform or one with outcome variance $O(D^{-2})$.  This corresponds to taking an equal a priori probability of $1/2$ for both outcomes.  The theorem then implies that the probability of distinguishing the measurement statistics from those that would be expected from the uniform distribution is at most $1/2+O(N/{D}^{1/4})$.  Therefore, $N=O(D^{1/4})$ samples are needed to distinguish the two possible models with probability greater than $1/2+\delta$ for any fixed $\delta>0$.  We then see from Definition 1 that almost all Hamiltonians drawn from this ensemble will equilibrate \change information theoretically with respect to $A$ and for any $t\ge t_{\rm eq}$ as $D\rightarrow \infty$.
\end{proofof}


\section{Unitary t-design Condition for Information Theoretic Equilibration}\label{t-design-argument}

We now discuss how the unitary, $C$, which transforms the eigenbasis of $H$ to that of the observable $A$ can be used to understand the equilibration properties of $H$.  Working in the eigenbasis of $A$,  we can write
$U(t) = C F(t) C^\dag$, where
$F(t)= \mbox{diag}[e^{-itE_a}]$, and $\{E_a\}_{a=1}^D$ are the energy eigenvalues of $H$.
Given an initial pure state $\rho(0)=\ketbra{x}{x}$, the measurement outcome probabilities can be written as
\beq
\prob{k|\rho(t)} =  \Tr{|k\rangle\!\langle k| C F(t)C^\dagger |x\rangle\! \langle x| C F(t)^\dagger C^\dagger }.\label{eq:probeq}
\eeq
Information theoretic equilibration \change follows from
$\var{ \prob{k|\rho(t)} }{k} \in O(D^{-2})$, which in turn requires that we know certain properties of $\prob{k | \rho(t)}^2$.
It is not difficult to see that $\prob{k|\rho(t)}$ and $\prob{k|\rho(t)}^2$ can be concisely represented by
\begin{align}
\prob{k|\rho(t)}&=\bra{L_2} C^{\otimes 2}\otimes {\bar{C}^{\otimes 2}}\ket{R_2(t)},\label{eq:22}\\
\prob{k|\rho(t)}^2&=\bra{L_4} C^{\otimes 4}\otimes {\bar{C}^{\otimes 4}}\ket{R_4(t)},\label{eq:44}
\end{align}
where $\bra{L_2}=\bra{k,x,x,k}$, $\bra{L_4}=\bra{k,x,k,x,x,k,x,k}$, and
\begin{align}
\ket{R_2(t)}&=\sum_{b,b'}e^{it(E_b-E_{b'})}\ket{b,b',b,b'},\\
\ket{R_4(t)}&=\sum_{b,b',d,d'}e^{it(E_b-E_{b'}+E_d-E_{d'})}\ket{b,b',d,d',b,b',d,d'}.
\end{align}
We refer to a term of the form $\bra{L_2} C^{\otimes 2} \otimes {\bar{C}^{\otimes 2}}\ket{b,b',b,b'}$ in the sum in~\eqref{eq:22} as a $(2,2)$--term because there are two basis change matrices acting on each factor space.  The analogous terms in~\eqref{eq:44} will be called $(4,4)$--terms; furthermore, $\mathbb{V}_{k}(\prob{k|\rho(t)})$ can be expressed as a $(4,4)$ polynomial (meaning that all terms in the expansion of the outcome variance are at most $(4,4)$ terms).  It can be easily checked using~(2) and Chebyshev's inequality that~3 and~(4) hold if the $(4,4)$-- and $(8,8)$--terms scale as $O(D^{-4})$ and $O(D^{-8})$ respectively.  \change This shows that we can reduce the question of whether typical Hamiltonians drawn from an ensemble equilibrates information theoretically with respect to an observable and time to a question about the properties of these terms.

The  scalings given in eqns.~3 and (4) are satisfied if the matrix elements of $C$, namely the Hamiltonian eigenvector components, satisfy a unitary $t$-design condition \cite{Dankertetal}, which means that these matrix elements  reproduce Haar-randomness for polynomials of degree at most $(t,t)$.
A similar connection was identified recently for subsystem equilibration in Refs.~\cite{Brandaoetal,MRA11,Znidaric11}, which required  a unitary $4$-design. 
\change For microcanonical equilibration, we must also ensure that $\mathbb{V}_{\ensemble}\{\mathbb{V}_k \prob{k|\rho(t)} \}$ is sufficiently small to imply a concentration of measure for $\mathbb{V}_{k}\{\prob{k|\rho(t)}\}$ via Chebyshev's inequality.  The resulting expression is an $(8,8)$--polynomial and hence a unitary $8$-design condition is sufficient to imply that  information theoretic equilibration with respect to local qubit measurements and $t\ge t_{\rm eq}$ is typical for individual systems from the ensemble (using  Theorem~1 and~Lemma~1).  

\section{Equilibration for Nearest Neighbor Hamiltonians}\label{AppNNH}
Previously, we showed numerically that random local Hamiltonians on a complete graph equilibrate information theoretically with respect to observables that are local rotations of $A=\sum_{j=0}^{D-1} a_j \ketbra{j}{j}$ for non--degenerate $a_j$ and $\ket{j}$ an eigenstate of $\sigma_z^{\otimes n}$.  Although some physical systems, such as the Bardeen--Cooper--Schrieffer Hamiltonian for low--temperature superconductivity~\cite{WBL02}, can be represented as random local Hamiltonians on a complete graph, many physically relevant Hamiltonians have interactions that are constrained to nearest neighbors.  We consider two relevant cases.  First, we consider random local Hamiltonians with nearest neighbor interactions on lines with periodic boundary conditions.  We then consider random local Hamiltonians on  square lattices with periodic boundary conditions.    In both cases, we see compelling numerical evidence for information theoretic equilibration with respect to $A=\sum_{j=0}^{D-1} a_j \ket{j}\!\bra{j}$ for non--degenerate $a_j$.

\emph{Random local Hamiltonians on a line}--  We will now consider Hamiltonians of the form
\begin{equation}\label{eq:1dHam}
H=\left(\sum_{i=1}^n\sum_{p} a_{i,p}\sigma_p^{(i)}+ \frac{1}{2}\sum_{\langle i,j\rangle}\sum_{p} b_{i,j,p,p'}\sigma_p^{(i)}\sigma_{p'}^{(j)}\right),
\end{equation}
where the sum over $\langle i, j\rangle$ refers to a sum over nearest neighbor $i$ and $j$ and $b_{i,j,p,p'}=b_{j,i,p',p}$.  In this case, such that only interactions between qubits $i$ and $j$ are permitted if $|i-j|=1$ or $|i-j|= n-1$.  Similarly to the RLH ensemble, we take $a_{i,p}$ and $b_{i,j,p,p'}$ to be Gaussian random variables with mean $0$ and variance $1$. 

Figure~\ref{fig:nnh} shows that typical members drawn from this constrained ensemble of random local Hamiltonians also achieve information theoretical equilibrium with respect to $A$.  We see from the results that the ensemble expectation of the outcome variance scales as $O(D^{-2})$ and that the ensemble variance of the outcome variance scales as $O(D^{-4.7})$.    We know from the results of Lemma~1 require that if the ensemble average and variance of $\mathbb{V}_k\{\prob{k|0;t} \}$ at most as $O(D^{-2})$ and $O(D^{-4})$ respectively in order to guarantee that a Hamiltonian sampled uniformly from the ensemble will, with high probability, equilibrate information theoretically with respect to the observable.  We therefore conclude from this data that information theoretic equilibration with respect to non--degenerate measurements in the eigenbasis of $\sigma_z^{\otimes n}$ is generic for members of this ensemble of random local Hamiltonians.

\begin{figure}[t!]
\includegraphics[width=\linewidth]{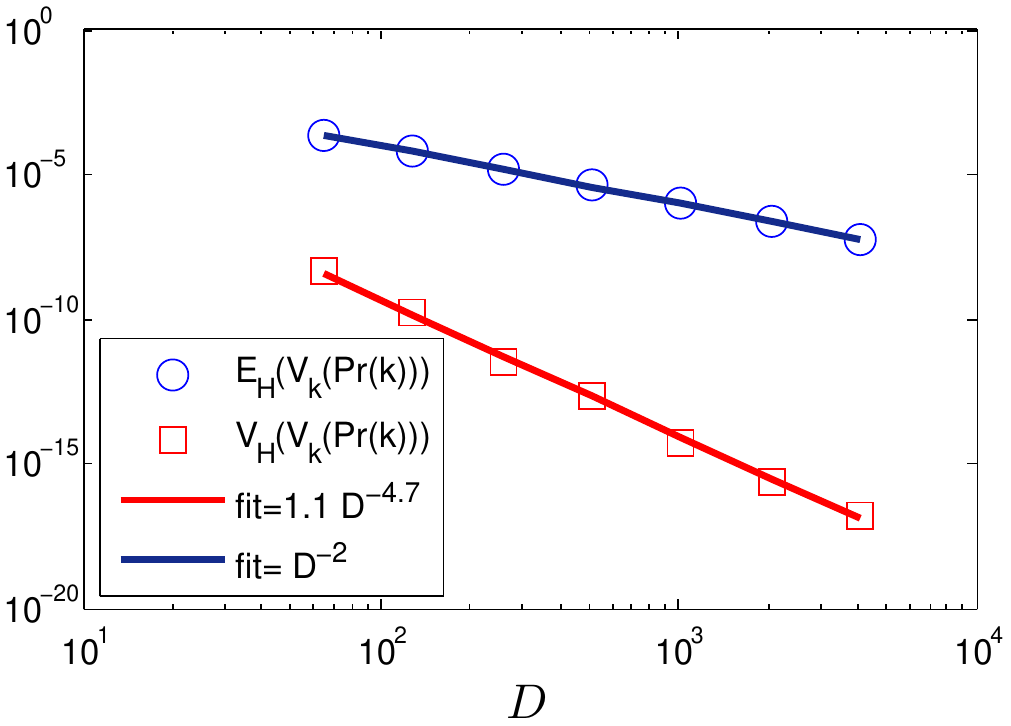}
\caption{This plot provides numerical estimates of the scaling of the ensemble average and variance of $\mathbb{V}_k\{\prob{k|0,t}\}$ for random local Hamiltonians with nearest neighbor interactions on a line with periodic boundary conditions and $t\ge t_{\rm eq}$.}\label{fig:nnh}
\end{figure}

\emph{Random local Hamiltonians on a square lattice}--  Next we examine the issue of whether typical members of the ensemble of random local Hamiltonians that are constrained to have only nearest neighbor interactions between qubits on a square lattice equilibrate information theoretically with respect to computational basis measurements.  The ensemble of Hamiltonians is similar to that in~\eqref{eq:1dHam} except now we permit two qubits to interact if the two qubits are adjacent vertices on a square lattice.  Note that we do not require that the overall shape of the lattice is a square.  Specifically, we consider lattices with a number of qubits $n=4,6,9,10,12,14$.  In the case of $n=4$ the lattice is uniquely a square of $2\times 2$ qubits.  In the case of $n=12$, there is an ambiguity in that the lattice can be expressed as an array of $4\times 3$ qubits or $2\times 6$ qubits.  We examine the former configuration because it is less like the $1$D case.

Figure~\ref{fig:nnh2} shows that Hamiltonians drawn uniformly from this ensemble of Hamiltonians constrained to nearest neighbor interactions on a square lattice, with high probability, equilibrate information theoretically with respect to computational basis measurements for exactly the same reasons as the one--dimensional case discussed above.  We also should note that although we have only studied equilibration with respect to a computational basis measurement, the results trivially also hold for local rotations of the computational basis because both the one-- and two--dimensional ensembles are invariant with respect to single qubit rotations of any and all qubits.

\begin{figure}[t!]
\includegraphics[width=0.945\linewidth]{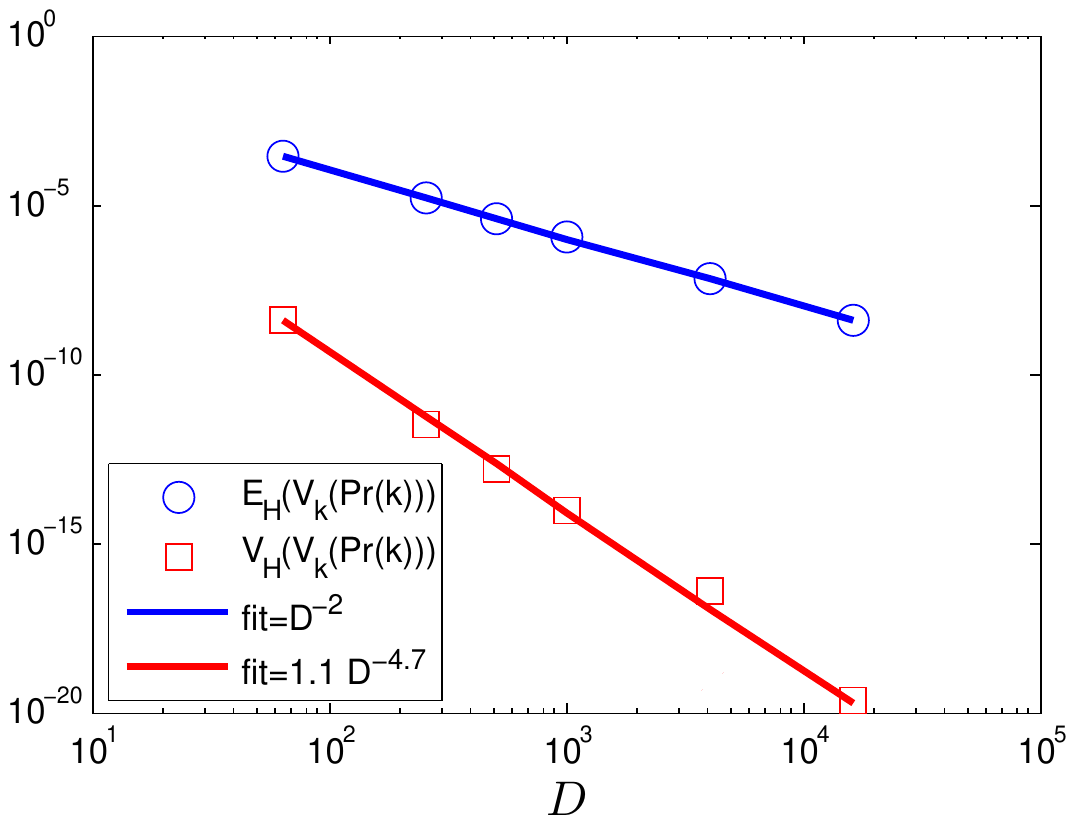}
\caption{This plot provides numerical estimates of the scaling of the ensemble average and variance of $\mathbb{V}_k\{\prob{k|0,t}\}$ for random local Hamiltonians with nearest neighbor interactions on a rectangular lattice with periodic boundary conditions and $t\ge t_{\rm eq}$.}  
\label{fig:nnh2}
\end{figure}

\section{Equilibration Time for a variant of the Heisenberg Model}
In the main body of the text, it was claimed that a variant Heisenberg model:
\begin{equation}\label{eq:genheis2}
H=\|H\|^{-1}\left(\sum_{i\le n/2} \sigma_z^{(i)}+\sum_{i> n/2} \sigma_x^{(i)}+ \sum_i \vec{\sigma}^{(i)}\cdot \vec{\sigma}^{(i+1)}\right).
\end{equation}
 has an equilibration time of $t_{\rm eq}\approx O(\log(D))$.  This fact can be seen in Figure~\ref{fig:teqH} where we plot the average probability, over $x$, of an initial eigenstate of $\ket{0}^{\otimes n}$ being measured in the state $\ket{x}\ne \ket{0}^{\otimes n}$ after evolution under the Hamiltonian for time $t$.  We see from the figure strong evidence for equilibration of the measurement outcome.  

The equilibration times do not scale as smoothly with $D$ as the data considered for the RLH ensemble.  The reason for this discrepancy is that the ratio of spins experiencing a magnetic field in the $X$ direction to those experiencing a field in the $Z$ direction varies with $n$.  If $n$ is even, then the ratio will be $1:1$; however, if $n$ is odd then there will be an excess of spins experiencing a transverse field in the $X$ direction.  This difference causes the equilibration times to vary with the parity of $n$.  We do see evidence though in Figure~\ref{fig:teqH} that $t_{\rm eq} \approx O(\log(D))$; although given the small range for the fit, the precise functional dependence of $t_{\rm eq}$  on $D$ is not certain.
\begin{figure}[t!]
\includegraphics[width=\linewidth]{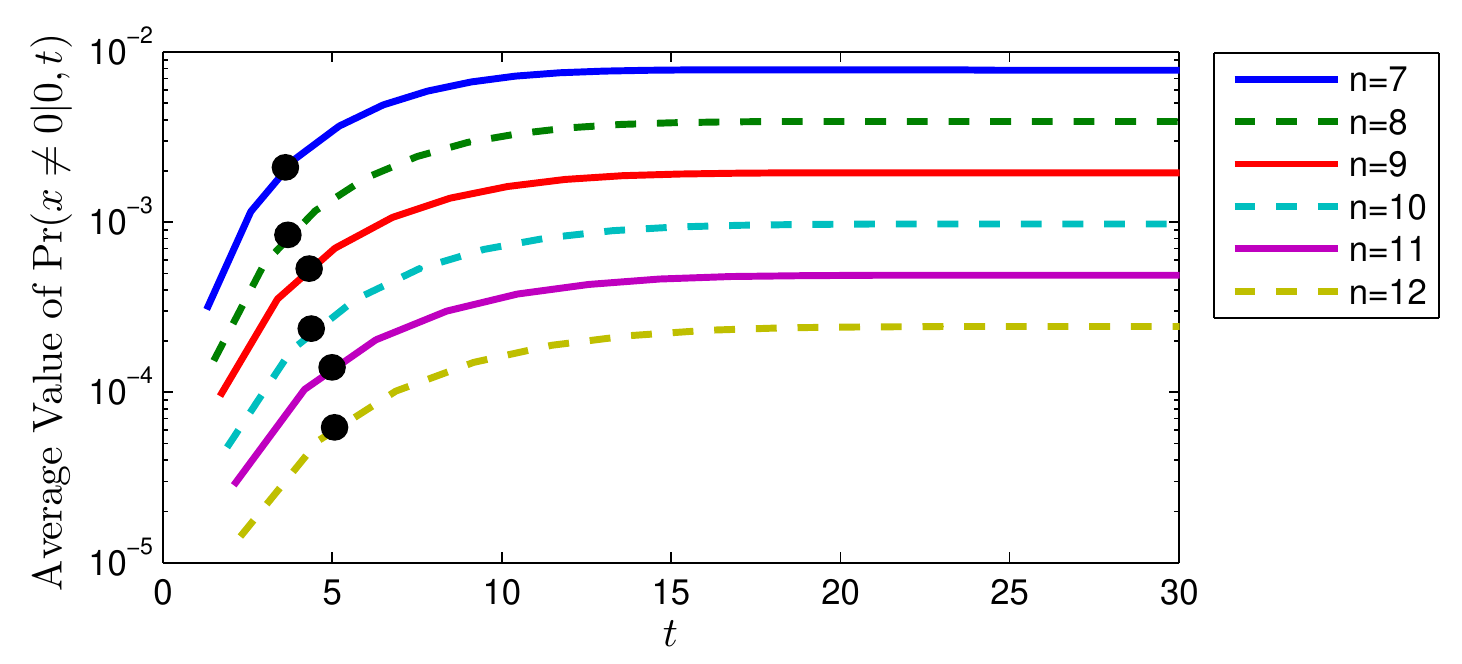}
\caption{This plot shows $\prob{k\ne 0|0,t}$ for the Hamiltonian given in~\eqref{eq:genheis2}.  The circles show $t_{\rm eq}$ for each $n$, which scale roughly as $O({\log(D)})$.}  
\label{fig:teqH}
\end{figure}

\section{Equilibration for a One--Body Quantum Chaotic System}

Finally we show that information theoretic equilibration of pure states occurs also for an individual system consisting of a one-body dynamical model associated with global classical chaos. In particular we consider a
 variant of the quantum kicked top,  described by the Floquet map~\cite{Haa06}:
\begin{equation}
U_F=e^{-i\bm{J}\bm{\tau}\bm{J}^T}e^{-i \mathbf{J}\cdot \mathbf{\alpha}},
\end{equation}
where $\mathbf{J}=[J_x, J_y, J_z]$ is a vector of angular momentum operators, $\tau$ is the moment of inertia tensor, which is diagonal with entries $[\tau_{xx},\tau_{yy},\tau_{zz}]$, $\bf{\alpha}=[\alpha_x,\alpha_y,\alpha_z]$ is a vector of kick--strengths and $\hbar=1$.  The dimension of the system is $D\!=\!2j\!+\!1$ where $j$ is the total angular momentum quantum number.
Fig.~\ref{fig:qkt} shows that for a set of chaotic parameters for $\lambda$ and $\tau$, roughly $10$ kicks (applications of $U_F$) are needed in order for $\var{\prob{k}}{k}\!=\!O(D^{-2})$, and hence for Theorem~1 to apply.

\begin{figure}[t!]
\includegraphics[width=\linewidth]{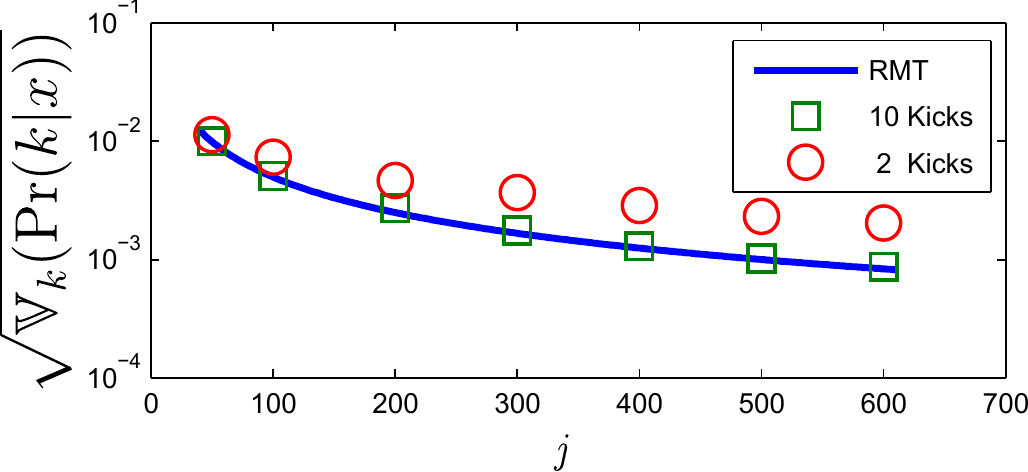}
\caption{The measurement outcome variance $\var{\prob{k|x}}{k}$ plotted as a function of system size,  both before ($2$ kicks) and after ($10$ kicks) the equilibration time.  For $10$ kicks, we see excellent agreement with the GUE prediction of $O(D^{-2})$, confirming information theoretic equilibration of a one-body system in the macroscopic limit in a regime of global chaos ($\bm{\alpha}=[1.1,1.0,1.0]$ and $\bm{\tau}={\rm diag}(10,0,1)$).\label{fig:qkt}}
\end{figure}

\section{Inferring Equilibration of RLH from $(4,4)$ Terms}
We showed previously that if the Hamiltonian is sufficiently complex, meaning that the change of basis matrix $C$ satisfies a condition similar to an $8$--design condition, then equilibration theoretic equilibration is generic for Hamiltonians drawn uniformly from the ensemble.  Here we use this insight to infer from the properties of the $(4,4)$ terms of RLH Hamiltonians that information theoretic equilibration is generic.

In order to show this, we need to estimate the ensemble means and variances of the dominant $(4,4)$ terms (in the limit of large $D$).  If the initial state preparation is $\ket{x}$ and measurement outcome $k$ is considered then the relevant sum in computing the value of $\prob{k|\rho(t)}$ (which is needed in the computation of $\mathbb{V}_k\{\prob{k|\rho(t)} \}$) is
\begin{widetext}
\begin{equation}
\prob{k|\rho(t)}^2=\sum_{b,b',d,d'}\bra{k,x,k,x,x,k,x,k}C^{\otimes 4}\otimes \bar{C}^{\otimes 4}{e^{it(E_b-E_{b'}+E_d-E_{d'})}}\ket{b,b',d,d',b,b',d,d'}.
\end{equation}
\end{widetext}
Here each summand is known as a $(4,4)$ term.  It is easy to then see that in the limit of large $t$ and under the assumption of non--degenerate Hamiltonians, that terms with $b=b'$ and $d=d'$ will dominate other terms due to the fact that no phase cancellation appears in the sum over such terms.  This means that the value of the outcome variance is dictated by the characteristic magnitude of such terms.  In this case, we do not need to compute the $(2,2)$--terms  because $\mathbb{E}_{\ensemble}(\prob{k|\rho(t)})=1/D$ trivially holds for RLH Hamiltonians.

\begin{figure}[t!]
\includegraphics[width=\linewidth]{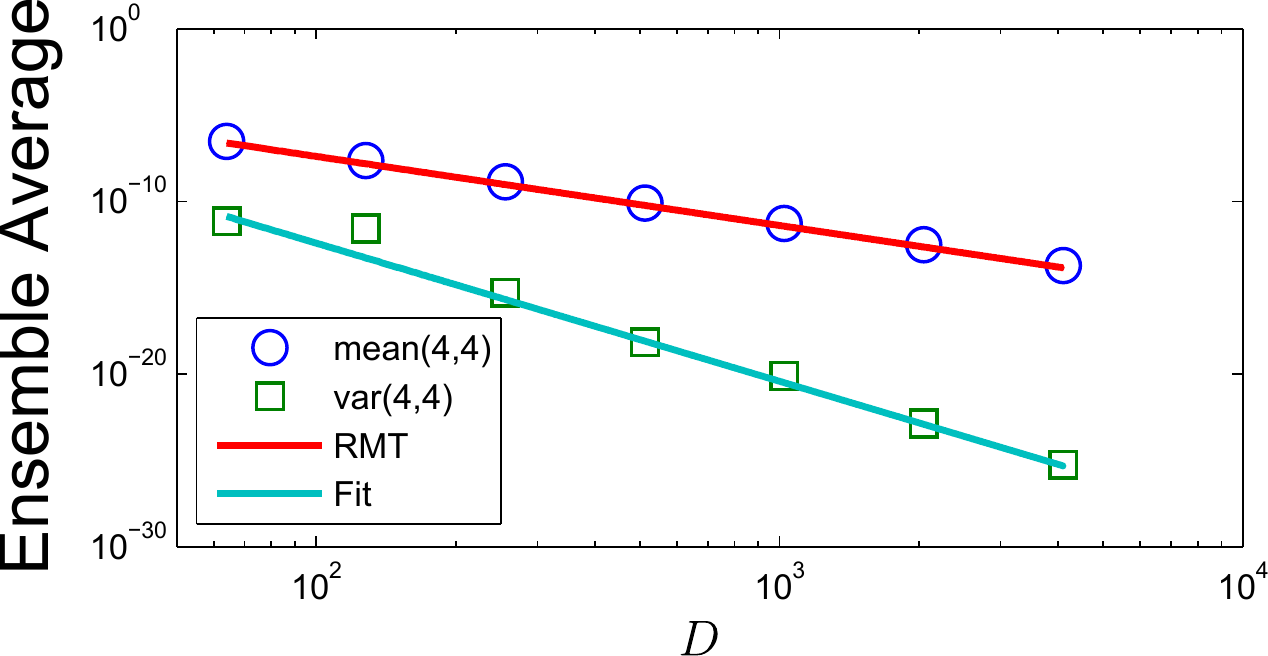}
\caption{Numerically computed expectation values over $k$ and $b$ for $x=0$ of the RLH ensemble average and variances of the $(4,4)$--terms that are dominant for large $t$, $\bra{L_4}C^{\otimes 4}\otimes \bar{C}^{\otimes 4}\ket{b,b,b,b,b,b,b,b}$, plotted as a function of $D=2^n$, and compared
to the scalings for GUE.  The variance of the $(4,4)$ terms scale as $3790/D^8$ for $D> 128$.  The results do not qualitatively change if $k$ and $b$ are fixed or if $b'\ne b$.\label{fig:RLH2}}
\end{figure}

Fig.~\ref{fig:RLH2} shows that the RLH average of the $(4,4)$--terms agrees with the GUE predictions of $O(D^{-4})$ scaling.   We also find that the variance of the $(4,4)$ terms is $O(D^{-8})$ which suggests that concentration of measure holds for the ensemble.  This shows that the equilibriation properties observed for the RLH ensemble can also be inferred from the properties of the change of basis matrix $C$.

\section{Proof of Theorem 2}\label{AppI}

We prove Theorem~2 in two steps:
\begin{enumerate}
\item We derive expressions for the Haar expectations of $\prob{k | \rho(t)}$ (see \eqref{prob defn} below) and the second and fourth moments.
    See Lemmas~\ref{C var}, \ref{C avg}, and~\ref{C 4th}.
\item We find the expectations of these expressions over the GUE eigenvalue distribution, and show that there exists a time $t_{\rm eq}(D)=O(D^{-1/6})$ such that for all $t> t_{\rm eq}(D)$, eqns. (3) and (4) hold.  See Section \ref{Appendix GUE} for these expectations.
\end{enumerate}
Since the calculations in this subsection require in depth knowledge of the properties of the Gaussian Unitary Ensemble (GUE), we will begin by giving a brief review of its properties \cite{Haa06,Meh04,Bal68}. GUE is the unique unitarily invariant distribution over Hermitian matrices $H$ that factorizes into a product of distributions each over an individual element of $H$.  In particular, each indepedent element of $H$ is an i.i.d.~Gaussian random variable.  A Hermitian matrix $H$ generated according to the GUE has diagonal elements $h_{aa}$ that are real valued random variables each with distribution $\mathcal{N}(0,\sigma^2)$, and off-diagonal elements $h_{ab}$ with real and imaginary parts that are random variables each with distribution $\mathcal{N}(0,\frac{1}{2} \sigma^2)$.
The variance $\sigma^2$ is a free parameter, which is closely related to the expected maximum energy eigenvalue as well as the ensemble average energy level spacing.

While many aspects of GUE have been shown to accurately model complex quantum systems, there are known limitations to using GUE as a model of such systems which deserve mention before we proceed.
First, the average level density, which takes the form of a semi-circle, and long-range spectral fluctuations for GUE are not good models for the corresponding properties of physically relevant Hamiltonian systems, even chaotic ones.  In particular, for most natural systems, such as those with only two particle interactions, the norm of the Hamiltonian scales polynomially with the number of particles \cite{LPS+09}.  However, the expected norm of a GUE Hamiltonian scales polynomially in the Hilbert space dimension $D$ \cite{Haa06}.
As we will see below, this has a large impact on what might be called the equilibration time for these dynamical systems, which should therefore be taken with a grain of salt.  Because of this, and in order to simplify calculations, we will follow the standard practice \cite{MRA11,Meh04} of taking the variance
$\sigma^2=\frac{1}{2}$.

The joint distribution over all elements of $H$ factorizes into a product of a distribution over eigenvectors, and a distribution over the energy eigenvalues of $H$ (see \cite{Meh04} Theorem 3.3.1, or \cite{Haa06} Chapter 4).  Further, the joint probability distribution over eigenvectors of $H$ is the same as that over the change of basis matrix $C$, namely the Haar measure on the unitary group $\mathbb{U}(D)$ \cite{Haa06,Meh04}.
Letting $C$ be the unitary which takes the eigenbasis of $H$ to that of $A$, and working in the eigenbasis of $A$, we can write
\beq
U(t) = C F(t) C^\dag,
\eeq
 where
$F(t):= \mbox{diag}[e^{-itE_a}]$, and $\{E_a\}_{a=1}^D$ are the energy eigenvalues of $H$.
We can therefore take separate expectations over eigenvectors and eigenvalues, namely, over the matrices $C$ and $F$.  $\prob{k | \rho(t)}$ can be expressed as
\beq\label{prob defn}
\prob{k | \rho(t)} =  \Tr{ |k\rangle \langle k| C F(t)C^\dagger |x\rangle \langle x| C F(t)^\dagger C^\dagger}.
\eeq
In the following we will write the expectation over the Haar measure on the unitary group $\mathbb{U}(D)$ of change of basis matrices $C$, as $\expect{ \cdot }{C}$, and $\expect{ \cdot }{\rm spec}$ for the expectation over the GUE eigenvalue distribution.

\subsection{Expectations over eigenvectors}\label{eigenvect exp}

\begin{lemma}\label{C var}
Take a non-degenerate observable $A$ acting on $\mathcal{H}=\mathbb{C}^D$, an initial pure state $\rho(0) = \ket{x}\!\bra{x}$ which is an eigenstate of $A$, and a unitary $U(t) = e^{-itH}$ where $H$ is drawn uniformly at random from GUE.
Then the variance of the measurement outcome probabilities over the Haar measure on the unitary group $\mathbb{U}(D)$ of change of basis matrices $C$ is given by:
\begin{align}
&\var { \prob{k | \rho(t)} }{C} \nonumber \\
&\qquad=
  \frac{1}{D^4}
\Big\{ D^2-2D+4 + (7-2D)\delta_{xk} \nonumber \\
&\qquad+ |\mu (t)|^2 (2 D\delta _{xk}-10 \delta _{xk}-2)  \nonumber\\
&\qquad+ \delta_{xk}|\mu (2 t)|^2 +2\delta_{xk} Re[\mu (t)^2 \mu (-2 t)] \Big\}+O(D^{-5}),
\end{align}
where we have defined $\mu(t)= \Tr{U(t)} = \Tr{F(t)}$
(this is often called the spectral form factor \cite{Haa06,Meh04}).
\end{lemma}

\begin{proof}
First, note that the squares of the outcome probabilities $\prob{k | \rho(t)}$ can be written in the form:
\begin{align}
& \Tr { |k\rangle \langle k| C F(t)C^\dagger |x\rangle \langle x| C F(t)^\dagger C^\dagger }^2    \nonumber \\
&\qquad= \sum_{s,s'} \langle s| k\rangle \langle k| C F(t) C^\dag | x\rangle \langle x| C F(t)^\dag C^\dag
|s \rangle   \nonumber \\
& ~~\qquad \times \ \langle s'| k \rangle \langle k | C F(t) C^\dag | x\rangle \langle x| C F(t)^\dag C^\dag    |s' \rangle \nonumber \\
&\qquad= \langle L_4 | C^{\otimes 4} \otimes \bar{C}^{\otimes 4}|R_4(t) \rangle,
\end{align}
where $\bar{C}$ is the complex conjugate of $C$ and
\begin{align}
\langle L_4 | &= \langle k,x,k,x, {x}, {k}, {x}, {k} |,  \nonumber\\
|R_4(t) \rangle &= \sum_{b,b',d,d'} e^{it(E_b -E_{b'}+ E_d-E_{d'})}  \nonumber \\
 &\times |b, b', d, d', {b},  {b'}, {d}, {d'} \rangle,\label{L and R def}
\end{align}
The expectation of the expression
$ C^{\otimes 4} \otimes \bar{C}^{\otimes 4} $ over Haar measure can be written as the projector onto the subspace spanned by the vectors
\beq
\ket{\Phi_\pi} = \left(V_\pi \otimes \id\right) \ket{\phi}_{1,5}\ket{\phi}_{2,6}\ket{\phi}_{3,7}\ket{\phi}_{4,8},
\eeq
where $\ket{\phi}_{ij}=  \sum_{a=1}^D \ket{a}_i\ket{a}_j$, and the index $\pi$ runs over the $4!$ permutations of the elements $\{1,2,3,4\}$, and $V_\pi$ is the unitary permutting the first four factor spaces according to $\pi$.
It was shown in \cite{Brandaoetal} that this projector is given by
\begin{align}
\expect{C^{\otimes 4} \otimes \bar{C}^{\otimes 4} }{C} = \sum_{\pi,\sigma} ( M^{-1})_{\pi,\sigma} \ket{\Phi_\pi}\!\bra{\Phi_\sigma},
\end{align}
where the matrix $M$ has components
$M_{\pi,\sigma} = \langle V_\pi | V_\sigma \rangle = \Tr {V_{\pi^{-1}} V_\sigma } = d^{l(\pi^{-1} \sigma)}$,
and $l(\sigma)$ is number of cycles in the cycle decomposition of the permutation $\pi^{-1} \sigma$.
We then have
\begin{align}\label{SumPermutationsM}
\expect{ \prob{k | \rho(t)}^2  }{C} =\sum_{\pi,\sigma} \langle L_4   \ket{\Phi_\pi}( M^{-1})_{\pi,\sigma} \bra{\Phi_\sigma} R_4(t) \rangle,
\end{align}
where the inner products $\bra{\Phi_\sigma} R_4(t) \rangle$ are given by:
\beqa
\bra{\Phi_{(1, 2, 3, 4)}} R_4(t) \rangle &=& |\mu(t)|^4, \nonumber \\
\bra{\Phi_{(1, 2, 4, 3)}} R_4(t) \rangle=\bra{\Phi_{(1, 3, 2, 4)}} R_4(t) \rangle &=& d|\mu(t)|^2, \nonumber\\
\bra{\Phi_{(2, 1, 3, 4)}} R_4(t) \rangle =
\bra{\Phi_{(4, 2, 3, 1)}} R_4(t) \rangle &=& d|\mu(t)|^2, \nonumber\\
\bra{\Phi_{(1, 3, 4, 2)}} R_4(t) \rangle=\bra{\Phi_{(1, 4, 2, 3)}} R_4(t) \rangle &=& |\mu(t)|^2 , \nonumber\\
\bra{\Phi_{(2, 3, 1, 4)}} R_4(t) \rangle= \bra{\Phi_{(2, 4, 3, 1)}} R_4(t) \rangle &=& |\mu(t)|^2 , \nonumber\\
\bra{\Phi_{(3, 1, 2, 4)}} R_4(t) \rangle= \bra{\Phi_{(3, 2, 4, 1)}} R_4(t) \rangle &=& |\mu(t)|^2 , \nonumber\\
 \bra{\Phi_{(4, 1, 3, 2)}} R_4(t) \rangle= \bra{\Phi_{(4, 2, 1, 3)}} R_4(t) \rangle &=& |\mu(t)|^2 , \nonumber\\
\bra{\Phi_{(3, 4, 1, 2)}} R_4(t) \rangle &=& |\mu(2t)|^2 ,\nonumber\\
\bra{\Phi_{(1, 4, 3, 2)}} R_4(t) \rangle &=& \mu(t)^2 \bar{\mu}(2t), \nonumber\\
\bra{\Phi_{(3, 2, 1, 4)}} R_4(t) \rangle &=& \mu(2t) \bar{\mu}(t)^2 , \nonumber\\
\bra{\Phi_{(2, 1, 4, 3)}} R_4(t) \rangle= \bra{\Phi_{(4, 3, 2, 1)}} R_4(t) \rangle &=& d^2 , \nonumber\\
\bra{\Phi_{(2, 3, 4, 1)}} R_4(t) \rangle= \bra{\Phi_{(3, 4, 2, 1)}} R_4(t) \rangle &=& d, \nonumber\\
\bra{\Phi_{(2, 4, 1, 3)}} R_4(t) \rangle= \bra{\Phi_{(3, 1, 4, 2)}} R_4(t) \rangle &=& d, \nonumber\\
\bra{\Phi_{(4, 1, 2, 3)}} R_4(t) \rangle= \bra{\Phi_{(4, 3, 1, 2)}} R_4(t) \rangle &=& d, \nonumber
\eeqa
Further, $\langle L_4   \ket{\Phi_{\pi}} = 1$ for
$\pi = (4, 3, 2, 1)$, $(4, 1, 2, 3)$, $(2, 3, 4, 1)$, $(2, 1, 4, 3)$, and $\langle L_4 \ket{\Phi_{\pi}}= \delta_{ik}$ for all other $\pi$.

Taking the sum of the above terms with $M^{-1}$ in
eqn. \eqref{SumPermutationsM}, we find:
\begin{align}
  &\expect{ \prob{k | \rho(t)}^2 }{C}  \nonumber \\
  &= \frac{1}{\alpha}
  \Big\{ |\mu(t)|^4 \left( (D^2 -D -2) \delta_{xk}+2 \right) \nonumber \\
 &+ |\mu(t)|^2 \left( -4D^2\!-\!12D\!-8\!+ (4D^3\!+\!8D^2\!+\!4D+8) \delta_{xk} \right) \nonumber \\
 &+ (D^2 -D -2) \delta_{xk}\left(|\mu (2 t)|^2+
 \bar{\mu}(2t)\mu(t)^2+ \mu(2t) \bar{\mu}(t)^2 \right) \nonumber\\
 &+2D^4+8D^3+6D^2 -(4D^3+12D^2)\delta_{xk}  \nonumber\\
 &+2|\mu(2t)|^2
 +2\bar{\mu}(2t) \mu(t)^2+2\mu(2t)\bar{\mu}(t)^2
    \Big\},
\end{align}
where $\alpha = { D^2 (D-1)(D+1) (D+2) (D+3)}$.
Using $\expect{ \prob{k | \rho(t)} }{C}$ from Lemma \ref{C avg},
we then have
\begin{align}
&\var{ \prob{k | \rho(t)} }{C} \nonumber  \\
&=\frac{1}{D^4}
\Big\{ D^2-2D+4 + (7-2D)\delta_{xk}\nonumber \\
&     + |\mu (t)|^2 (2 D\delta_{xk}-10 \delta_{xk}-2) \nonumber \\
&     +\delta_{xk}|\mu(2t)|^2 +2\delta_{xk} Re[\mu(t)^2 \mu (-2 t)] \Big\}+O(D^{-5}).
\end{align}
\end{proof}

\begin{lemma}\label{C avg}
Take a non-degenerate observable $A$ acting on $\mathcal{H}=\mathbb{C}^D$, an initial pure state $\rho(0) = \ket{x}\!\bra{x}$ which is an eigenstate of $A$, and a unitary $U(t) = e^{-itH}$ where $H$ is drawn uniformly at random from GUE.
Then the expectation of the measurement outcome probabilities for $C$ distributed according to the Haar measure on the unitary group $\mathbb{U}(D)$ is given by:
\beq
 \expect { \prob{k |\rho(t)} }{C}
=
\frac{D - \delta_{xk} +  |\mu(t)|^2 \left( \delta_{xk}- \frac{1}{D} \right)}{D^2-1}.
\eeq
\end{lemma}
\proof
This expectation can be calculated in a similar but simpler fashion as the variance in the previous lemma, so we leave the proof as an exercise.
\qed

 \begin{lemma}\label{C 4th}
Take a non-degenerate observable $A$ acting on $\mathcal{H}=\mathbb{C}^D$, an initial pure state $\rho(0) = \ketbra{x}{x}$ which is an eigenstate of $A$, and a unitary $U(t) = e^{-itH}$ where $H$ is drawn uniformly at random from the GUE.
Then the fourth moment of the measurement outcome probabilities over the Haar measure on the unitary group $\mathbb{U}(D)$ of change of basis matrices $C$ is given by:
\begin{align}
&\expect{\prob{k | \rho(t)}^4 }{C} \nonumber   \\
&\qquad\leq \frac{1}{D^8} \sum_{\pi}    \bra{\Phi_\pi} R_8(t) \rangle +
\sum_{\pi, \sigma}   B_{\pi,\sigma} \bra{\Phi_\sigma} R_8(t) \rangle,
\end{align}
where $B$ is a matrix with components $\leq O(D^{-9})$.
\end{lemma}
\begin{proof}
The fourth power of the outcome probabilities can be written in the form:
\begin{align}
  &\Tr {|k\rangle\! \langle k| C F(t)C^\dagger |x\rangle\! \langle x| C F(t)^\dagger C^\dagger }^4 \nonumber \\
  &\qquad\qquad=   \langle L_8 | C^{\otimes 8} \otimes \bar{C}^{\otimes 8}|R_8(t) \rangle,
\end{align}
where $\langle L_8 |$ and $|R_8(t) \rangle$ are defined analogously to \eqref{L and R def}, but with twice the number of tensor factors.
Further, the average of the expression
$ C^{\otimes 8} \otimes \bar{C}^{\otimes 8} $ over Haar measure can be written as the projector onto the subspace spanned by the vectors
\beq
\ket{\Phi_\pi} = \left(V_\pi \otimes \id\right) \ket{\phi}_{1,9}\ket{\phi}_{2,10}\ldots\ket{\phi}_{8,16},
\eeq
where the index $\pi$ runs over the $8!$ permutations of the elements $\{1,2,\ldots,8\}$, and $V_\pi$ is the unitary permuting the first eight factor spaces according to $\pi$.

We will now determine the asymptotic scaling of the following expression with $D$:
\begin{align}
&\expect{ \prob{k | \rho(t)}^4  }{C}\nonumber\\
&\qquad = \sum_{\pi,\sigma} \langle L_8  \ket{\Phi_\pi}( M^{-1})_{\pi,\sigma} \bra{\Phi_\sigma} R_8(t) \rangle\label{4th power sum},
\end{align}
by finding an approximation for $M^{-1}$.
Recall that the matrix $M$ has components
$M_{\pi,\sigma} = \langle V_\pi | V_\sigma \rangle = {\rm Tr}[V_{\pi^{-1}} V_\sigma]  = D^{l(\pi^{-1} \sigma)}$,
where $l(\sigma)$ is number of cycles in the cycle decomposition of the permutation $\pi^{-1} \sigma$.
It is clear that the diagonal components of $M$ are all equal to
$D^8$, and all other components of $M$ are strictly $<D^8$, as the identity permutation is the only one with $8$ cycles.
Letting
$
A := \id -  {M}/{D^8},
$
we have
\beq
(\id -A)  \left(\sum_{k=0}^N A^k\right) = \left(\sum_{k=0}^N A^k\right)(\id -A)  = \id - A^{N+1},
\eeq
which converges to $\id$ in the limit $N \rightarrow \infty$ if and only if $||A||_{\rm op}<1$.
Because $A$ is of fixed size $8! \times 8!$ and all its elements are $< 1/D$, we have $||A||_{\rm op}<1$ for $D$ sufficiently large.
Therefore,
$
(\id - A)^{-1} = \id + A + O(D^{-2})
$,
which implies that
\beq
M^{-1} = \frac{\id}{D^8} + O(D^{-9}).
\eeq

Next, it is not difficult to see that all components of $\langle L_8  \ket{\Phi_\pi}$ are equal to $1$ or $\delta_{xk}$.
Therefore, we can split the sum in \eqref{4th power sum} as
\begin{align}
& \expect{  \prob{k | \rho(t)}^4  }{C}  \nonumber\\
&\qquad\leq  \frac{1}{D^8} \sum_{\pi}    \bra{\Phi_\pi} R_8(t) \rangle +
\sum_{\pi, \sigma}   B_{\pi,\sigma} \bra{\Phi_\sigma} R_8(t) \rangle,
\end{align}
where we have used that $\langle L_8  \ket{\Phi_\pi} \leq 1$ for all $\pi$, and $B = M^{-1} - \frac{\id}{D^8}$ has all components $\leq O(D^{-9})$.
\end{proof}

\subsection{Expectations over GUE eigenvalues}\label{Appendix GUE}
Our next  step towards proving Theorem~2 is to find the GUE average of the the spectral form factor $\mu(t)$, which appears in Lemmas~\ref{C var} and~\ref{C avg}.
The joint distribution over the un-ordered energy eigenvalues $\{E_a\}_{a=1}^D$ of a GUE matrix is given by (see \cite{Meh04} Theorem 3.3.1, or \cite{Haa06} Chapter 4):
\beq\label{GUE distribution}
P(\{E_a\}^D)  = \frac{1}{C_D} \prod_{1=a<b}^{D} (E_a-E_b)^{2}\
{\rm exp}\Big[ - \frac{1}{2 \sigma^2} \sum_{a=1}^D E_a^2 \Big],
\eeq
where we define $C_D  := (2\pi)^{D/2} \sigma^{D^2} \prod_{j=1}^{D} j! $.
Integration of $P(\{E_a\}^D) $ over $D-m$ variables gives the m-point correlation function (\cite{Meh04} 6.1.1)
\beq
R_m(E_1,\ldots,E_m) = \chi \int_{-\infty}^{\infty} P(\{E_a\}^D) dx_{m+1} \cdots dx_{D},
\eeq
where $\chi = {D!}/{(D-m)!}$.
By \cite{Meh04} Theorem 5.1.4, this can be written as:
\beq\label{defn of R_n}
R_m(E_1,\ldots,E_m) = {\rm det} [K_D(E_i,E_j)]_{i,j=1,\ldots,m},
\eeq
where
$K_D(E_i,E_j) =\sum_{k=0}^{D-1} \phi_k(E_i) \phi_k(E_j)$,
and the $\phi_k(x)$ are the harmonic oscillator wave-functions
$\phi_j(x) =  \left(2^j j! \sqrt{\pi}\right)^{-1/2} e^{- x^2/2}H_j(x)$,
with $H_j(x)$ the Hermite polynomials (see \cite{Meh04} section 6.2).

\subsubsection{Calculation of $\expect{|\mu(t)|^2 }{\rm spec}$}\label{form factor calc}

Our next goal is to evaluate the expression
 $\Gamma_{}(t,D):=\expect{|\mu(t)|^2 }{\rm spec}=
 \expect{   \left| \Tr{e^{-itH}} \right|^2 }{\rm spec}$.
 Expanding out the trace and grouping terms, we have
\begin{align}
\Gamma_{}(t,D) &= \expect{ \big|{\rm Tr}\big[e^{-itH}\big]\big|^2 }{\rm spec} \nonumber \\
&=
\int_{-\infty}^{\infty} \left|\sum_{l=1}^D e^{it E_l} \right|^2 P(\{E_a\}^D) \prod_{j=1}^D dE_j  \nonumber \\
&=
D \int_{-\infty}^{\infty} P(\{E_a\}^D) \prod_{j=1}^D dE_j \nonumber \\
 &\qquad+
\int_{-\infty}^{\infty} \sum_{l \neq m}^D e^{it (E_l-E_m)} P(\{E_a\}^D) \prod_{j=1}^D dE_j.
\end{align}
Using the invariance of the joint distribution $P(\{E_a\}^D)$ under permutations of the energies, and the definition of the 2-point correlation function, this becomes
\beq
\Gamma_{}(t,D)= D +
\int_{-\infty}^{\infty} e^{it (E_2-E_1)} R_2(E_1,E_2) dE_1 dE_2.
\eeq
It is well known that for large $D$ the function $K_D(E,E)$ follows the so called Wigner semi-circle law \cite{Meh04}
(with $\sigma^2 =1/2$):
\beq
K_D(E,E) \simeq
 \frac{1}{\pi}\sqrt{2D - E^2} \Theta( \sqrt{2D} - |E|),
\eeq
where $\Theta(x)$ is the Heavyside step function.  For large $D$ we then have:
\begin{align}\label{Gamma}
\Gamma_{}(t,D) &=  D +   2D \frac{J_1\left(\sqrt{2D}  t \right)^2}{t^2} \nonumber \\&\qquad-
 (\sqrt{2D}-t/2)\Theta( 2\sqrt{2D}-t),
\end{align}
where $J_1(x)$ is the first Bessel function of the first kind.
{It should also be noted that, the limit of the function $K_D(E_1,E_2)$ is generally taken by rescaling (often called `unfolding' - see \cite{GMW98} section III.A.1) the energy level density as well as the energies by the local mean spacing \cite{Haa06, Meh04}.  In taking this limit and the integrals for $\Gamma_{}(t,D)$ we have followed the the procedure of \cite{Meh04} Appendices 10, and 11.}

\subsubsection{Equilibration time}\label{Equil time}

Now that we have the GUE expectation of the spectral form factor, we can use Lemma \ref{C avg}, to find the full expectation over eigenvectors and eigenvalues of the outcome probabilities:
\begin{align}
&\expect{ \prob{k | \rho(t)} }{\rm spec,C} \nonumber\\
 &\qquad= \frac{D
- \frac{1}{D} \Gamma_{}(t,D)
+ \delta_{xk} \left( \Gamma_{}(t,D)  -1 \right) }{D^2-1},
\end{align}
Notice that if $\Gamma_{}(t,D) = O(D)$, then
\beq\label{expected prob}
\expect{ \prob{k | \rho(t)} }{\rm spec,C} = \frac{D+\delta_{xk}O(D)}{D^2-1}+O(D^{-2}).
\eeq
In particular, if $D$ is large, then the GUE expectation of the probability distribution from \eqref{expected prob} is essentially the uniform distribution.
We therefore take the \emph{equilibration time} $t_{\rm eq}(D)$ to be defined by the condition
 \beq\label{equil defn}
t_{\rm eq}(D) := \{ T \ |\ \Gamma_{}(t,D)= O(D), \forall~t>T \}.
\eeq
It is clear that the second term in \eqref{Gamma} satisfies:
\beq\label{eq:teqgue}
\lim_{t \rightarrow \infty} 2D \frac{J_1\!\left(\sqrt{2D}t \right)^2}{t^2} = 0.
\eeq
This shows that there exists a finite time $t_{\rm eq}(D)$ such that
$\Gamma_{}(t,D) =O(D)$ for all $t> t_{\rm eq}(D)$.

In order to get a sense of the equilibration time (keeping in mind our previous comments on GUE energy spectra), note that the condition on the equilibration time \eqref{equil defn} is essentially that
\beq\label{equil cond2}
2D \frac{J_1\left(\sqrt{2D}  t \right)^2}{t^2}  = O(D).
\eeq
Using the fact that for $x \gg 3/4$ we can approximate
$J_1(x) \simeq \sqrt{2/\pi x}\cos[x-3\pi/4]$, it the follows that the equilibration time is
\beq\label{equil time1}
t_{\rm eq}(D) = O(D^{-1/6}).
\eeq

\subsubsection{GUE expectation of first and second power}\label{sect GUE mean and var}

Our next step towards a proof of Theorem~2 involves evaluating the expectation values over the spectrum that are needed to prove
eqns. 
 (3) and (4) for GUE Hamiltonians.
First, from the previous section we have that there exists a finite time $t_{\rm eq}(D)$ such that
$\Gamma_{}(t,D)=O(D)$ for all $t> t_{\rm eq}(D)$.
Using Lemma \ref{C avg}, and the above spectral expectation, we have
\beq\label{eq:guemean}
\expect{ \prob{k | \rho(t)} }{spec,C} = \frac{D+\delta_{xk}O(D)}{D^2-1}+O(D^{-2}),
\eeq
which proves that for all $t > t_{\rm eq}(D)$, eqn. (3) holds for GUE Hamiltonians.

Next, Lemma \ref{C var} implies that
\begin{align}
&\expect{\var{\prob{k | \rho(t)}  }{C} }{\rm spec} \nonumber \\
&\qquad= \frac{1}{D^4}
\Big\{ D^2-2D+4 + (7-2D)\delta_{xk} \nonumber \\
&\qquad~~+ \Gamma_{}(t,D) (2 D\delta _{xk}-10 \delta _{xk}-2)+ \delta_{xk}\Gamma_{}(2t,D) \nonumber \\
&\qquad~~+  2\delta_{xk} \expect{ Re[\mu (t)^2 \mu (-2 t)] }{\rm spec} \Big\}+O(D^{-5}).
\end{align}
In the next section we will show that for $t > t_{\rm eq}(D)$,
$\expect{ Re[\mu (t)^2 \mu (-2 t)] }{\rm spec} = O(D)$.
Using this along with $\Gamma_{}(t,D) =O(D)$ in the above, we find
\beq\label{eq:guevar}
\expect{ \var{ \prob{k | \rho(t)}  }{C} }{\rm spec}=O(D^{-2}),
\eeq
which proves that for all $t > t_{\rm eq}(D)$, eqn. (3) holds.

\subsubsection{Bounding $\expect{ \mu(t)^2 \mu(-2t) }{\rm spec}$}\label{sect horrible}

In order to simplify the derivation of our upper bounds, we define
\beqa
\Delta_{}(t,D)&:=&\expect{ \mu (t)^2 \mu (-2 t) }{\rm spec}\nonumber  \\
&=&\int_{-\infty}^{\infty} \sum_{i,j,k} e^{it(E_i+E_j -2E_k)} P(\{E_a\}^D) \prod_{l=1}^D dE_l,\nonumber\\
\label{splitting sums}
\eeqa
We can expand the triple sum into three parts depending on whether $i,j,k$ are: (i) all distinct, (ii) only two equal, or (iii) all equal, which upon integration over the remaining energies, give terms of the form:
\begin{enumerate}[{\normalfont (i)}]
\vspace{-\topsep}
\item $R_3(E_1,E_2,E_3) e^{it(E_1+E_2 -2E_3)}$,
\vspace{-\topsep}
\item  $R_2(E_1,E_2) e^{2it(E_1-E_2)} + 2 R_2(E_1,E_2) e^{it(E_1-E_2)}$,
\vspace{-\topsep}
\item  $D$.
\vspace{-\topsep}
\end{enumerate}
By expanding $R_3$ and $R_2$, we find products of integrals of terms of the form $K_D(E_i,E_i) e^{it_iE_i}$, as well as $K_D(E_i,E_j)^2 e^{i(t_iE_i+t_jE_j)}$, and
\beq\label{bounding this}
K_D(E_1,E_2)K_D(E_2,E_3)K_D(E_3,E_1)e^{it(E_1+E_2-2E_3)}.
\eeq
We will bound the magnitude of these terms in a similar fashion as was done in \cite{MRA11}.
Note that the integral of \eqref{bounding this} is just
\beq\label{traces}
\Tr{P e^{iE_1} P e^{iE_2} P e^{-2iE_3}},
\eeq
where $P = \sum_{k=0}^{D-1} | \phi_k \rangle\! \langle \phi_k|$ is the projector onto the D-dimensional lower-energy subspace spanned by the harmonic oscillator wave-functions, and $X$ is the position operator.  Using the Cauchy-Schwartz inequality twice on~\eqref{traces}, we find:
\begin{align}
\left| \Tr{(P e^{iE_1} P e^{iE_2}) P e^{-2iE_3}} \right| &\leq
\sqrt{\sqrt{\Tr{P}\Tr{P}} \Tr{P}} \nonumber \\ &\leq D.
\end{align}
A similar argument also shows that terms with integrands of the form
$K_D(E_i,E_j)^2 e^{i(t_iE_i+t_jE_j)}$ are also bounded by $D$.
Using these bounds, we then have
\begin{align}
&\left|\Delta_{}(t,D)\right| \leq 6D \nonumber\\
&+ \left( \int_{-\infty}^{\infty} K_D(E,E) e^{itE} \right)^2
\left(  \int_{-\infty}^{\infty}  K_D(E,E) e^{-2itE}\right) \nonumber \\
&+\!D\!\left(  \int_{-\infty}^{\infty}  \! K_D(E,E) e^{-2itE}\right)\!
\!+\!2D\left(  \int_{-\infty}^{\infty}  \! K_D(E,E) e^{itE}\right)
\nonumber \\
&+ \left| \int_{-\infty}^{\infty}  K_D(E,E)e^{2itE} \right|^2
+  2\left|  \int_{-\infty}^{\infty}  K_D(E,E)  e^{itE}\right|^2.
\end{align}
Next recall that
\beq
\int_{-\infty}^{\infty} K_D(E,E) e^{itE} =
\sqrt{2D} \frac{J_1(\sqrt{2D}\ t)}{t},
\eeq
which approaches $0$ as $t \rightarrow \infty$.  The equilibration condition \eqref{equil defn} requires that $\sqrt{2D}J_1(\sqrt{2D}\ t)/t = O(\sqrt{D})$ for $t > t_{\rm eq}(D)$, and if this is satisfied, then for all $t > t_{\rm eq}$, we have
$\left|\Delta_{}(t,D)\right| \leq O(D)$.

\subsubsection{GUE expectation of fourth power}\label{final 4th}

The final quantity that we need to prove Theorem~2 is the asymptotic scaling of the GUE average of the
fourth moment of $\prob{k|\rho(t)}$ where $\rho(t)=e^{-iHt}\ketbra{x}{x}e^{iHt}$ for $\ket{x}$ an eigenvector of the observable.  This quantity shows a concentration of measure for the outcome variance of $\prob{k|\rho(t)}$, which will allow us to conclude that individual Hamiltonians drawn from the GUE will information theoretically equilibrate with high probability.
Specifically, we show that for $t>t_{\rm eq}(D)$ we have that
$\expect{ \prob{k | \rho(t)}^4 }{spec,C}\! =\! O(D^{-4})$.
Recalling the form of $\expect{ \prob{k | \rho(t)}^4 }{C}$ from Lemma \ref{C 4th},
we see that it is sufficient to show that for $t>t_{\rm eq}(D)$, we have
$\expect{ \bra{\Phi_\sigma} R_8(t) \rangle }{\rm spec} \leq O(D^{4})$, for all permutations $\sigma$.

Calculating some explicit examples of
\begin{align}
&\bra{\Phi_\sigma} R_8(t) \rangle \nonumber\\
&\qquad= \sum_{\substack{a,a',b,b',\\c,c',d,d'}}
e^{it(E_a -E_{a'} + E_b -E_{b'}+ E_c -E_{c'}+ E_d-E_{d'})} \nonumber \\
&\qquad\times \langle a,a',b,b',c,c',d,d' |V_\sigma | a,a',b,b',c,c',d,d'\rangle,\nonumber\\
\end{align}
we see that these are of the general form
\beq
D^{a}\mu(f_1t)\mu(g_1t)^*,\ldots, \mu(f_4t)\mu(g_4t)^*,
 \eeq
 where $f_j,g_j \in \{0,1,2,3,4\}$, and
 \beq\label{constraint}
 a+ \sum_{j=1}^4 f_j + \sum_{j=1}^4 g_j = 8,
  \eeq
  and if $f_j =0$ then the corresponding $\mu(f_jt)$ does not appear in the product (and similarly for $g_j$).
For example, there is a $|\mu(t)|^4$ term arising from $\sigma = \id$, and a $D^4$ term arising from $\sigma = (12)(34)(56)(78)$.
The expectations
$\expect{ D^{a}\mu(f_1t)\mu(g_1t)^*,\ldots \mu(f_4t)\mu(g_4t)^* }{\rm spec}$ can then be bounded in a similar fashion to Appendix \ref{sect horrible}.  In particular, we can expand the sums in the product of the terms
$\mu(f_jt)= \Tr{e^{if_jtH}}$ into various parts depending on which indices
are equal or not equal, just as was done for expression \eqref{splitting sums}.  These will then give a constant $D^b$ factor, and various combinations of integrals of m-point correlation functions.  Each of the integrals with 2-point or higher order correlation functions can be bounded by $D$, just as was done for eqn. \eqref{bounding this}.  We will then be left with various powers of integrals of the form
$\int_{-\infty}^{\infty} K_D(E,E) e^{itE}$
which as we have seen approach $0$ as $t \rightarrow \infty$, and so are irrelevant for the $t > t_{\rm eq}$ regime.  The only remaining question then is power of the  constant $D^b$ factor for each term.  It is not difficult to see that a power of $D$ arises from each pairing $\mu(f_jt)\mu(g_kt)^*$ with $f_j=g_k$.  For example, the expectation of the term $|\mu(t)|^4$ gives a contribution of $D^2$ (this is in fact the infinite time limit).  From the constraint \eqref{constraint}, it is not difficult to see that $D^4$ is the highest power of $D$ which can arise.
This shows that for $t > t_{\rm eq}$
\beq\label{eq:4thpower}
\expect{  {\rm Pr}(k|\rho(t))^4  }{{\rm spec},C} = O(D^{-4}).
\eeq

Putting all of the above together, we have:
\begin{proofof}{Theorem~2}
\change
Expanding the outcome and ensemble variances and using~\eqref{eq:guevar} and~\eqref{eq:4thpower}, we find that, for $t> t_{\rm eq}(D)$,
\begin{align}
&\var{\var{\prob{k|\rho(t)}}{k} }{\ensemble} \nonumber\\
&=\!
D^{-2} \sum_{j,k}\! \expect{ \prob{k|\rho(t)}^2 \prob{j|\rho(t)}^2 }{\ensemble}\! + O(D^{-4}).
\end{align}
Then using the Cauchy-Schwartz inequality for expectations and~\eqref{eq:guemean}, we have:
\beq
\sum_{j,k} \expect{\prob{k|\rho(t)}^2 \prob{j|\rho(t)}^2 }{\ensemble} \leq
 O(D^{-2}).
\eeq
This proves that $\var{\var{\prob{k|\rho(t)}}{k}}{\ensemble}= O(D^{-4})$.
Chebyshev's inequality then implies (in the same fashion as~\eqref{eq:chebyprob})  that
$\var{\prob{k|\rho(t)}}{k} \in O(D^{-2})$ with high probability over $\ensemble$.   Theorem 1 then implies that $O(D^{1/4})$ samples are required to distinguish the $\prob{k|\rho(t)}$ from the uniform distribution for $t>t_{\rm eq}(D)=O(D^{-1/6})$, which was shown in \eqref{equil time1}.  Since a particular $GUE$ Hamiltonian is specified using  $O(D^{2})$, even drawing a random Hamiltonian from GUE requires $O(\rm{poly}(D))$ arithmetic operations.  Therefore it follows from Definition 1 and Lemma 1 that almost all GUE Hamiltonians equilibrate information theoretically for $t>t_{\rm eq}$ with high probability.
\end{proofof}

\noindent
\emph{Acknowledgements.}
We thank  Fernando Brand\~{a}o, Carl Caves, David Cory, Patrick Hayden, Daniel Gottesman, and Victor Veitch for insightful comments.
We acknowledge funding from CIFAR, Ontario ERA, NSERC, US ARO/DTO and thank the Perimeter Institute for Theoretical Physics where this work was brought to completion.  



\end{document}